\documentclass[10pt,journal,compsoc]{IEEEtran}
\usepackage{graphicx}
\usepackage{cite}
\usepackage{algorithm}
\usepackage{amsmath,amsthm,amssymb}
\graphicspath{ {images/} }
\usepackage{multirow}
\usepackage{epstopdf}
\usepackage{color}
\usepackage{epsfig}
\usepackage{xfrac}
\usepackage{bbm}
\usepackage{filecontents}
\usepackage{csquotes}
\usepackage{footnote}
\makesavenoteenv{tabular}
\makesavenoteenv{table}
\usepackage{url}
\newtheorem{theorem}{Theorem}

\newtheorem{remark}[theorem]{Remark}
\theoremstyle{definition}
\newtheorem{definition}{Definition}[section]
\usepackage{pifont}
\usepackage{xfrac}
\usepackage{pbox}
\usepackage{epstopdf}
\usepackage[flushleft]{threeparttable}

\ifCLASSINFOpdf
\else
\fi

\begin{document}
	%
	\title{A Role-Based Encryption Scheme for Securing Outsourced Cloud Data in a Multi-Organization Context}
	
	\author{Nazatul~H.~Sultan, ~\IEEEmembership{Student~Member,~IEEE}, Vijay Varadharajan, ~\IEEEmembership{Senior~Member,~IEEE}, Lan Zhou and Ferdous~A.~Barbhuiya, ~\IEEEmembership{Member,~IEEE}, 
		\thanks{Nazatul H. Sultan and Ferdous A. Barbhuiya are with the Department
			of Computer Science and Engineering, Indian Institute of Information Technology Guwahati (IIITG), Assam, India.\protect\\ 
			E-mail: [nazatul, ferdous]@iiitg.ac.in}\\
		\thanks{ Vijay Varadharajan is with the Department of Computing, Faculty of Engineering and Built Environment, Global Innovation Chair Professor at The University of Newcastle, Callaghan, Australia. \protect\\  
			E-mail: Vijay.Varadharajan@newcastle.edu.au}
		
		\thanks{Lan Zhou is with the Advanced Cyber Security Engineering Research Centre (ACSRC), The University of Newcastle, Callaghan, Australia and with the Amazon, Seattle, USA \protect\\  
			E-mail: lanz.cn@gmail.com }
	}

	%
\IEEEtitleabstractindextext{
\begin{abstract}
Role-Based Access Control (RBAC) is a popular model which maps roles to access permissions for resources and then roles to the users to provide access control. Role-Based Encryption (RBE) is a cryptographic form of RBAC model that integrates traditional RBAC with the cryptographic encryption method, where RBAC access policies are embedded in encrypted data itself so that any user holding a qualified role can access the data by decrypting it. However, the existing RBE schemes have been focusing on the single-organization cloud storage system, where the stored data can be accessed by users of the same organization. This paper presents a novel RBE scheme with efficient user revocation for the multi-organization cloud storage system, where the data from multiple independent organizations are stored and can be accessed by the authorized users from any other organization. Additionally, an outsourced decryption mechanism is introduced which enables the users to delegate expensive cryptographic operations to the cloud, thereby reducing the overhead on the end-users. Security and performance analyses of the proposed scheme demonstrate that it is provably secure against Chosen Plaintext Attack and can be useful for practical applications due to its low computation overhead.
	\end{abstract}
	
	}

	\maketitle 
	\IEEEdisplaynontitleabstractindextext
	\IEEEpeerreviewmaketitle
	\IEEEraisesectionheading{\section{Introduction}
	\label{intro}}
	\IEEEPARstart{P}{ublic} cloud storage have already become popular with end-users such as individuals and organizations. Some popular public cloud storage includes Microsoft Azure Storage Service \cite{Azure}, Amazon S3 \cite{AmazonS3}, and Google Cloud Storage \cite{GoogleCloudStorage}. The public cloud storage provides the data owners\footnote{Data owners are the end-users who own data.} to reduce their investment costs like building their own storage infrastructures. It also enables ubiquitous data access through the Internet without worrying about management and maintenance of the outsourced data \cite{Armbrust2010}.
	\par 
	Although the benefits of public cloud storage are significant, there is some reluctance among the data owners to outsource their data in the storage servers of the public cloud due to their concern for data security and privacy \cite{Kamara2010, Shin2017}. As the public cloud is formed by one or more cloud storage servers which are often distributed geographically in different locations, data owners do not know for certain where their data are stored. There is a strong perception that end-users lose control over their data once it is uploaded to the cloud storage \cite{Zhou2013}. There is more reluctance, especially if the data is sensitive. In \cite{McAfee}, McAfee reported that $84\%$ of the outsourced data are sensitive or confidential such as health-related data, financial documents, and personal photographs. This concern arises as the service provider might have the ability to access the sensitive information of the data owners for various motivations. For example, if the outsourced data contain data owners' health records, the service provider has the potential to sell these data to health insurance agencies for financial gain \cite{Verizon}. Hence, it is critical to preserve the confidentiality of outsourced data so that no malicious entity, including the service provider, has the ability to access the data without proper authorization. In order to control access of the data stored in a public cloud, suitable access control policies and mechanisms are required. The access policies must restrict data access to anyone other than those intended by the data owners.

	\par 
	In large scale systems such as enterprise, Role-Based Access Control (RBAC) has been widely used for access control, as it provides flexible security management. For instance, it allows access control to be managed at a level that corresponds closely to the organisation's policy and structure. Roles in organizations are mapped to access permissions and users are mapped to appropriate roles. In general, users are assigned membership to the roles based on their responsibilities and qualification in the organisation, and permissions are assigned to qualified roles instead of individual users. Moreover, in RBAC, a role can inherit permissions from other roles, therefore there can be a hierarchical structure of roles. This is one of the major advantages of the RBAC system. However, in the traditional RBAC \cite{Ferraiolo1999, Ferraiolo1992}, a service provider defines and enforces access policies on behalf of the data owners. The data owners have to, therefore, assume that the service provider is trusted to prevent unauthorised users from accessing their data. However, in an untrusted environment like a cloud environment, the service provider may give access privileges to unauthorized users for its own benefit, which can lead to reduced level of trust on the part of the data owners on the service provider when it comes to defining and enforcing access policy. This, in turn, makes the traditional RBAC mechanism less suitable for public cloud storage system. 
	\par
     Role-Based Encryption (RBE) is a cryptographic data access control method which is designed by combining the traditional RBAC model with the encryption method. It enables data owners to define and enforce RBAC access policy on the encrypted data itself \cite{Zhou2011, Zhou2013}. It thus reduces dependency on the service provider for defining and enforcement of RBAC access policy while giving full control to the data owners. In RBE, roles are organized to form a hierarchy, and each role is associated with a group of users who possess that role. Data are encrypted in such a way that any user who holds the required roles can derive the appropriate keys to decrypt the data. As the RBE is based on the RBAC model, the inheritance properties of RBE makes it more suitable for large organizations with complex hierarchical structures. 
	
	\par 
	There have been a few RBE schemes proposed over the recent years such as \cite{Zhou2011, Zhou2013, Zhu2013, Perez2017}. All of the schemes are designed for a single-organization public cloud storage system, where the outsourced data are hosted in the public cloud by the data owners of an organization, and the data are accessible to the users of the same organization only. In such a case, a single authority maintains all the roles and role hierarchies of that organization. The authority is responsible for assigning roles and corresponding role-keys to the users according to their access privileges in the role hierarchy. Later, the users can access data using their assigned roles, and role-keys will be used to decrypt the appropriate encrypted data. As a result, a user must hold a role in the organization to access data stored by that organization. 
	\par
	
	In many practical scenarios such as consortium\footnote{Consortium is an association where multiple organizations come together to share data to achieve some common goals.}, the data owners from multiple organizations outsource their data to the public cloud, and the data are shared with the users of other organizations. It might also happen that two organizations wish to work together on a collaborative project requiring users from these two organizations to share data in a secure manner. This is a more challenging multi-organization public cloud storage scenario. In such scenarios, the existing schemes \cite{Zhou2013, Zhu2013, Perez2017} fail as an authority in one organization will not be able to establish roles for the users in the other organization(s) for the lack of trust among the organizations. One approach to solving this problem is to define a fresh RBE system for the consortium creating separate roles and role hierarchies for the users of the consortium from different organizations. However, this creates a practical challenge as it is difficult to define the authority who can manage this consortium role system, and to which organization should this authority report to.

	\par 
	
    In this paper, we have developed a novel multi-organization RBE scheme, which addresses the above challenge of sharing outsourced encrypted data between multiple organizations.
    Our multi-organization RBE scheme will achieve the sharing of encrypted data across multiple organizations in such a way that only users with appropriate roles belonging to different organizations are able to decrypt the data. Hence, the proposed scheme is suitable for secure data sharing among several organizations such as in a consortium or in collaborative projects, where the ability to decrypt and access the data from different organizations is only possible if the appropriate policies specified by those organizations are satisfied. This has been achieved without creating a fresh consortium RBE system as mentioned earlier.
    
    The realization of the proposed multi-organization RBE scheme has been achieved using a hybrid cloud architecture comprising a private cloud and public cloud{\footnote {In a recent reference \cite{McAfee}, McAfee reported that 59\% of the organizations are using the hybrid cloud model, and its popularity is growing at an increasing rate among the organizations.}}. The private cloud is used to store only the sensitive information such as user and organization related secrets, and the public cloud is used to store the actual data in encrypted form. Users who wish to share or access the data only interact with the public cloud; there is no access for public users to access the private cloud, which greatly reduces the attack surface for the private cloud. This architecture not only dispels the organisation’s concerns about risks of leaking sensitive organization's information, but also takes full advantage of public cloud’s power to securely store large volume of data. 
    
    The noteworthy contributions of the scheme proposed in this paper are as follows:

	\begin{enumerate}
		\item [i)] A single-organization role-based encryption scheme referred to as SO-RBE is proposed. SO-RBE achieves access control over encrypted data in a single-organization cloud storage system, where the outsourced (encrypted) data are hosted by an organization. The encrypted data can be decrypted and accessed by the users of same organization who satisfy appropriate policies.

		\item [ii)] A separate multi-organization role-based encryption scheme referred to as MO-RBE is proposed by extending SO-RBE. MO-RBE enables a user of one organization to access data from another {\em without possessing any role from that organization}, provided the user is authorized. That is, the existing roles of the user from his/her own organization are alone sufficient to access data of the other organizations if the user is authorized (satisfies the appropriate access policies needed to decrypt the data).	
		
		\item [iii)] An efficient user revocation mechanism is introduced, which will work on both SO-RBE and MO-RBE, for revoking users from the system without the need to perform computationally expensive operations.
		\item [iv)] An outsourced decryption method is also introduced for outsourcing of the computation intensive cryptographic operations to the cloud without disclosing confidential information to the service provider. 	
		\item [v)] A formal security analysis of the proposed scheme is provided which shows that the proposed scheme is provably secure against \emph{Chosen Plaintext Attack} (CPA) under the standard cryptographic assumptions. 	
		\item [vi)] Moreover, performance analysis of the proposed scheme demonstrates that the proposed scheme is efficient to be used in practical applications.

	\end{enumerate}
	\par 
	To the best of our knowledge, this is the first RBE scheme that addresses access control over encrypted data for public cloud storage system in a multi-organization context.
	\par
	Unless stated otherwise, we refer the public cloud storage system as the cloud in the rest of this paper.
	\par 
	The rest of this paper is organized as follows. Related works are presented in Section \ref{related_works}. Section \ref{preliminaries} gives a brief overview of role hierarchy, properties of the bilinear map, and complexity assumptions that will be used throughout this paper. The system model, framework, security assumptions and security model are presented in Section \ref{proposed_model}. Section \ref{proposed_scheme} describes the proposed scheme in detail. The security and performance analysis details are given in Section \ref{analysis} and finally, Section \ref{conclusion} concludes this paper.
	
	\section{Related Works}
	\label{related_works}
	
\subsection{Hierarchical Key Management (HKM) Schemes} \label{HKA}
	Access control using hierarchical key management (HKM) method has been studied in the early 1980s. In \cite{Akl1983}, Akl \emph{et al.} presented the first cryptographic hierarchical access control technique to solve the hierarchical multi-level security problem. In this scheme, the users are grouped into disjoint sets (or classes) and form a hierarchical structure of classes. Each class is assigned with a unique encryption key and a public parameter in such a way that a higher-level class can derive encryption keys of any lower-level classes using its encryption key and some public parameters. Later several other HKM schemes have been proposed using different techniques, e.g. \cite{Chang2004, Atallah2009, Lin2011, Tang2016b, Chen2017, Pareek2018}. Recently, in \cite{Tang2016b}, Tang \emph{et al.} presented a HKM scheme based on linear geometry to provide flexible and fine-grained access control in cloud storage systems. In this scheme, any class in the hierarchy can derive encryption key using inner product of its public vector with the private vector of its ancestor class. In \cite{Chen2017}, Chen \emph{et al.} proposed another HKM scheme to support the assignment of dynamic reading and writing privileges. However, the main drawback of the HKM schemes is the high complexity for setting up the encryption keys for a large set of users \cite{Zhou2013}. Also, user revocation is a challenging task, as all the encryption keys that are known to the revoked users, and their related public parameters need to change per user revocation which may incur a high overhead on the system. 

	\subsection{Hierarchical ID based Encryption Scheme (HIBE)}
	An alternative approach for the management of keys is Hierarchical ID-based Encryption (HIBE) such as \cite{Gentry2002, Boneh2005}. In these schemes, a user can decrypt an encrypted data using the private key associated with his/her identity if and only if the data is encrypted using any descendant identity of the hierarchy (i.e., tree). That is, the user cannot access any data which are encrypted using the ancestor identities or any other identity of the hierarchy. As such, the HIBE schemes can enforce RBAC access policy in encrypted data by associating the users with leaf nodes and roles with non-leaf nodes in the hierarchy. However, in HIBE schemes, the size of an identity increases with depth of the hierarchy. In addition, the identity of a node must be a subset of its ancestor node so that its ancestor node can derive this node’s private key for decryption. Therefore, this node cannot be assigned as a descendant node of another node in the hierarchy tree unless the identity of the other role is also the super set of this node’s identity.
	
	\subsection{Attribute based Encryption (ABE) Schemes}
	
	The first attribute-based encryption (ABE) scheme was proposed in \cite{Goyal2006} based on the work in \cite{Sahai2005}, and some other ABE schemes have been proposed afterwards. In these schemes, data is encrypted using a set of attributes, and users who have the private keys associated with these attributes can decrypt the data. These works have provided an alternative approach to secure the data stored in a distributed environment using a different access control mechanism, such as \cite{Yub}.  In \cite{Zhou2011}, Zhou et al. have shown that an ABE scheme can be used to enforce RBAC policies. However, in that approach, the size of user key is not constant, and the revocation of a user will result in a key update of all the other users of the same role.  \cite{Zhu2013} also investigated the solutions of using ABE scheme in RBAC model. However their solution only maps the attributes to the role level in RBAC, and they assumed that the RBAC system itself would determine the user membership. There have also been other works based on variants of ABE such as \cite{Bethencourt2007, Goyal2006, Chase2007, Yang2014, Chuangui2018, FUGKEAW2018}. However, all these works  \cite{Bethencourt2007, Goyal2006, Chase2007, Yang2014, Chuangui2018, FUGKEAW2018} cannot address role hierarchy and inheritance properties. 
	
	\subsection{Role based Encryption (RBE)}
	
	A role based encryption (RBE) scheme integrates the RBAC access control model with cryptographic encryption techniques to enforce RBAC access policies on encrypted data. This enables the data to be encrypted to specific roles. Any user who holds the required role(s) and satisfies the associated RBAC access policies can access the data by decrypting it. It thus provides a secure solution for outsourcing data to the cloud storage servers while giving access to authorized users by defining and enforcing role based access policy on the encrypted data itself. 
	\par 
	In \cite{Zhou2011}, Zhou \emph{et al.} proposed the first RBE scheme for data sharing in an untrusted hybrid cloud environment. In this scheme, the ciphertexts and secret keys of the users are of constant size. This scheme also offers efficient user revocation capability. A complete design and implementation of the proposed scheme with a real world application has been described in \cite{Zhou2013}. In \cite{Zhu2013}, Zhu \emph{et al.} proposed another RBE scheme that increases ciphertext size linearly with the number of roles. In the scheme \cite{Zhu2013}, user revocation is performed during encryption process by the data owner that embeds revoke user information on the encrypted data itself. Hence data owner must know all the revoked users' information prior to the data encryption phase. In \cite{Perez2017}, Perez \emph{et al.} proposed a data-centric access control mechanism for cloud storage systems using the concept of proxy re-encryption technique. In \cite{Perez2017}, the data owner outsources data before encrypting it using the identity-based encryption technique proposed in \cite{Wang2016}. To share data with the authorized users, the data owner generates re-encryption keys based on the RBAC access policies and keeps the re-encryption keys with the ciphertexts. When an authorized user accesses the ciphertext, the service provider re-encrypts the ciphertext using the re-encryption keys based on a RBAC access policy. However, none of these schemes \cite{Zhou2011, Zhou2013, Zhu2013, Perez2017} support data access control in multi-organization cloud storage systems.
\par
	
	In this paper, we have used the RBE scheme described in \cite{Zhou2011, Zhou2013} as the basis for developing our RBE mechanism for multi-organization cloud environment, which has significant advantages compared to the ABE. It has the natural ability to enforce RBAC policies on encrypted data stored in the cloud with an efficient user revocation. A distinct advantage is that RBE is able to deal with role hierarchies, whereby roles inherit permissions from other roles. This is particularly significant as role based systems has been widely used in many large scale commercial systems providing flexible access control management corresponding closely to the organisation’s policy and structure. 
	\par
	Similarly in a multi-authority context, a multi-organization scheme based on ABE is not able to deal with the hierarchical structure of the organization in an efficient manner. On the other hand, a multi-organization scheme based on RBE such as the one proposed in this paper naturally fits with the organizational structure and is able to capture effectively the role hierarchies and inheritance properties. Furthermore, several of the multi-authority ABE schemes such as \cite{Chase2007, Yang2014, Yang2013} and \cite{Xue2017} require a trusted centralized/global authority to manage the set of attribute authorities. In this sense, these schemes represent a single-organization cloud environment. However, there are a few multi-authority ABE schemes such as \cite{WANG2011, Wan2012, Wei2019} can provide role hierarchy and inheritance property. These schemes use a hierarchical key generation method, where one top-level domain authority can generate keys for the lower-level domain authorities or users. But, \cite{WANG2011, Wan2012, Wei2019} assume that same root domain authority manages all the attributes of the system \cite{Perez2017}. As such, these schemes also represent a single-organization cloud environment. 
	
	

	\par 
		
		

	\begin{figure}[b]
		\centering
		\scalebox{3.5}{\includegraphics[width=2.5cm, height=1.8cm]{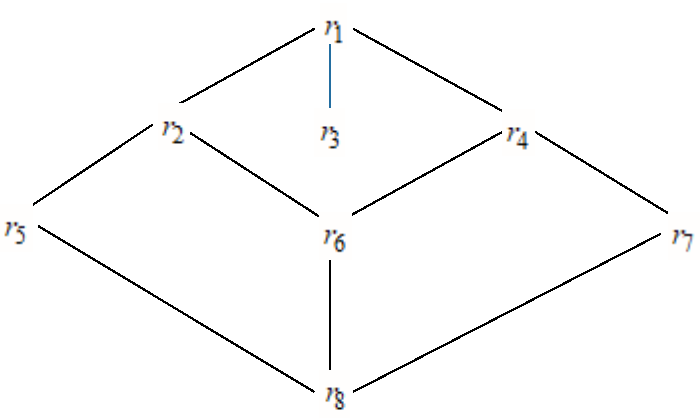}}
		\caption{Sample Role Hierarchy}
		\label{hierarchy}
	\end{figure}

\begin{figure*}[h]
	\centering
	\includegraphics[scale=.88]{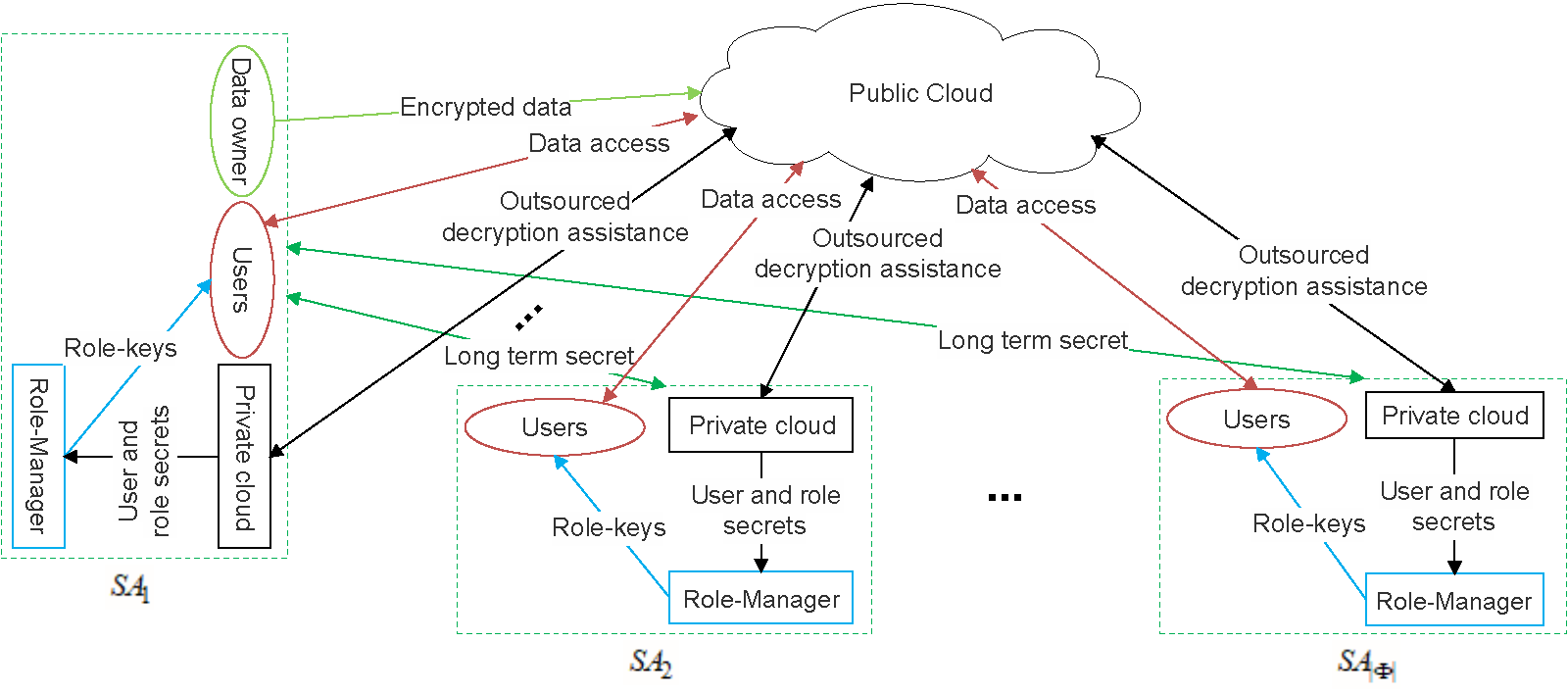}
	\caption{Proposed System Model}
	\label{model}
\end{figure*}


	\section{Preliminaries}
	\label{preliminaries}
	This section introduces the concept of role hierarchy, properties of bilinear mapping as well as standard complexity assumptions.

	\subsection{Role Hierarchy}
	\label{role_hierarchy}
	In the proposed scheme roles are organized into a hierarchy where ancestor roles can inherit access privileges of its descendant roles. Figure \ref{hierarchy} shows a sample role hierarchy. The following notations are used to define a role hierarchy:
	\begin{itemize}
		\item $\Psi$: set of all roles in the role hierarchy. For example, $\Psi= \{r_1, r_2, r_3, r_4, r_5, r_6, r_7, r_8\}$
		\item $\mathbb{A}_{r_i}$: ancestor set of the role $r_i$. For example, $\mathbb{A}_{r_8}= \{r_1, r_2, r_4, r_5, r_6, r_7, r_8\}, \mathbb{A}_{r_5}=\{r_1, r_2, r_5\}$ and $\mathbb{A}_{r_6}=\{r_1, r_2, r_4, r_6\}$.
		\item $\mathbb{D}_{r_i}$: descendant set of the role $r_i$. For example, $\mathbb{D}_{r_5}= \{r_8\}$, $\mathbb{D}_{r_1}= \{r_2, r_3, r_4, r_5, r_6, r_7, r_8\}$, $\mathbb{D}_{r_3}= \emptyset$ and $\mathbb{D}_{r_4}= \{r_6, r_7, r_8\}$ .
		\item $\bar{\mathbb{A}}_{r_i}$: set of roles not in $\mathbb{A}_{r_i}$, i.e., $\left(\Psi\setminus\mathbb{A}_{r_i}\right)$. For example, $\bar{\mathbb{A}}_{r_8}= \{r_3\}$, $\bar{\mathbb{A}}_{r_5}= \{r_3, r_4, r_6, r_7, r_8\}$ and $\bar{\mathbb{A}}_{r_6}= \{r_3, r_5, r_7, r_8\}$.
		\item $\bar{\mathbb{D}}_{r_i}$: set of roles not in $\mathbb{D}_{r_i}$, i.e., $\left(\Psi\setminus \mathbb{D}_{r_i}\right)$. For example, $\bar{\mathbb{D}}_{r_5}= \{r_1, r_2, r_3, r_4, r_5, r_6, r_7\}$, $\bar{\mathbb{D}}_{r_1}= \{r_1\}$ and $\bar{\mathbb{D}}_{r_4}= \{r_1, r_2, r_3, r_4, r_5\}$.
	\end{itemize}
	\subsection{Bilinear Map}
	Let $\mathbb{G}_1$ and $\mathbb{G}_T$ be two cyclic multiplicative groups. Let, $g$ be a generator of $\mathbb{G}_1$. The bilinear map $\hat{e}: \mathbb{G}_1\times \mathbb{G}_1\rightarrow \mathbb{G}_T$ has the following properties:
	\begin{itemize}
		\item \textit{Bilinear}: $\hat{e}(g^a, g^b)= \hat{e}(g, g)^{ab}$ $\forall g$ and $\forall a, b\in \mathbb{Z}_q^*$
		\item \textit{Non-degenerate}: $\hat{e}(g, g)\neq{1}$
		\item \textit{Computable}: There exists an efficiently computable algorithm for computing $\hat{e}(g, g)$ for all $g\in \mathbb{G}_1$
	\end{itemize}
	
	\subsection{Complexity Assumptions}
	\label{assumption}
	\begin{enumerate}
		\item \textbf{Decisional Bilinear Diffie-Hellman (DBDH) Assumption}: Let $\hat{e}: \mathbb{G}_1\times \mathbb{G}_1\rightarrow \mathbb{G}_T$ be an efficiently computable bilinear map. The DBDH assumption states that no probabilistic polynomial-time algorithm is able to distinguish the tuples $\left<g^a, g^b, g^c, \hat{e}(g, g)^{abc}\right>$ and $\left<g^a, g^b, g^c, \hat{e}(g, g)^{z}\right>$ with non-negligible advantage, where $a, b, c, z\in \mathbb{Z}_q^*$ and $g\in \mathbb{G}_1$.
		\item \textbf{Decisional Modified Bilinear Diffie-Hellman (MDBDH) Assumption} \cite{Sahai2005}: Let $\hat{e}: \mathbb{G}_1\times \mathbb{G}_1\rightarrow \mathbb{G}_T$ be an efficiently computable bilinear map. The MDBDH assumption states that no probabilistic polynomial-time algorithm is able to distinguish the tuples $\left<g^a, g^b, g^c, \hat{e}(g, g)^{\frac{ab}{c}}\right>$ and $\left<g^a, g^b, g^c, \hat{e}(g, g)^{z}\right>$ with non-negligible advantage, where $a, b, c, z\in \mathbb{Z}_q^*$ and $g\in \mathbb{G}_1$.
	\end{enumerate}
	\begin{table}[h]
		\tabcolsep 1.0pt
		\centering
		\caption{NOTATIONS}
		\begin{tabular}{p{1.3cm}p{6.9cm}}
			\hline
			Notation  & Description
			\\[0.5ex]    \hline
			$q$ & a large prime number   \\\hline
			$\mathbb{G}_1, \mathbb{G}_T$ & two cyclic multiplicative groups of order $q$   \\\hline
			$H_1(.)$ & hash function $H_1: \{0, 1\}^w\rightarrow \mathbb{Z}_q^*$\\\hline 
			$\Phi$ & set of organization in the system\\\hline
			$|\Phi|$ & total number of organizations in the system\\\hline
			$\Psi_k$ & set of roles managed by $k^{th}$ organization\\\hline
			$\mathtt{SA_k}$& $k^{th}$ system administrator\\\hline
			$\mathtt{RM_{r^k_i}}$ & a role-manager of $k^{th}$ organization which manages role $r^k_i$\\\hline	
			$r^k_{i}$ & $i^{th}$ role which is managed by $k^{th}$ organization\\\hline
			${ID}^k_{u}$ & unique identity of the $u^{th}$ user registered with $k^{th}$ organization\\\hline
			$\mathtt{RK}^{k}_{i, u}$& role-key of the role $r^k_i$ issued to the user ${ID}^k_{u}$ \\\hline
		\end{tabular}
		\label{notation}
	\end{table}
	
	\section{Proposed Model}
	\label{proposed_model}
	In this section, system model of the proposed scheme is presented along with framework, security assumptions and security model.
	\subsection{System Model}
	\label{cloud_model}
	Figure \ref{model} shows the proposed system model. It consists of six entities, namely system administrator, role-manager, private cloud, public cloud, data owners, and users which are described next. The notations used in this paper are shown in Table \ref{notation}.	
	\begin{itemize}
	\item \textit{System Administrator} (SA): It is a certified authority of an organization. It is responsible for managing role hierarchy of an organization. It generates system parameters including master secret and system public parameter. It is responsible for issuing a pair of private and public keys for each registered users. System administrator also issues role secret for each role in the organization to the respective role-manager. It also manages private cloud of an organization. System administrator keeps a part of the master secret and user secrets in the private cloud. Further, it revokes users when required. Moreover, each system administrator shares its \emph{long term secret} with the other system administrators in the cloud system so that its users can access data from the other organizations. 
	
	\item \textit{Private Clouds}: It is formed by the internal cloud storage servers of an organization which is managed by the system administrator. The responsibility of the private cloud is to keep confidential information of the organization. Private cloud mainly stores a part of the master secrets of the system administrator and user secrets. It uses the known master secrets while assisting the public cloud during outsourced decryption process, and it uses user secrets while assisting the role-manager for generating role-keys in the key generation process. Private cloud provides interfaces to the public cloud and to the role-managers only.
	
	\item \textit{Role-Manager}: It is an entity which is responsible for managing roles. Each role-manager in an organization manages each of its corresponding role(s) and the associated users in that role(s). The role-manager assigns roles to each registered users and issues role-keys related with the assigned roles for them. During the role-key generation process, the role-manager interacts with the private cloud to compute role-keys for the users. 

		
		\item \textit{Public Cloud}: It is formed by one or more cloud storage servers which are managed by a third-party service provider known as cloud service provider. The main responsibility of the public cloud is to store data owners' outsourced (encrypted) data in its cloud storage servers. The other responsibility is to perform outsourced decryption. 
		\item  \textit{Data Owner}: It is an entity who owns the data. A data owner encrypts data using RBAC access policy before outsourcing to the public cloud. 
		
		\item \textit{User}: It is an entity who uses the outsourced data. Each user needs to register with a system administrator. For each registered user, the role managers assign roles in form of the role-keys based on their profiles and responsibilities. The registered users also receive private keys from the system administrator.
	\end{itemize}
	
	\subsection{Framework}  
	\label{frameworks}
	The proposed scheme consists of the following phases and algorithms. 	
	\begin{enumerate}
		\item \textit{System Initialization}: This phase initializes the system, and it is initiated by a system administrator. It consists of the following algorithm.
		\begin{itemize}
			\item \textsc{Init} $((g^{\delta_k}, \mathtt{PK_{SA_k}}, \mathtt{MSK_{SA_k}})\leftarrow 1^{\Lambda_k})$: It takes security parameter $\Lambda_k$ as input. It outputs  master secret $\mathtt{MSK_{SA_k}}$, public parameter $\mathtt{PK_{SA_k}}$ and another secret $g^{\delta_k}$.
		\end{itemize}
		\item \textit{Manage Role}: This phase is initiated by a system administrator to generate secret role parameter, role public keys and role secrets. It comprises the following algorithm.
		\begin{itemize}
			\item \textsc{RoleParaGen} $((\mathtt{RP_{SA_k}}, \{\mathbb{PK}^k_{i}\}_{r^k_{ i}\in \Psi_k},\\ \mathtt{(RS_{r^k_i})_{\forall r^k_{ i}\in \Psi_k}})\leftarrow (\mathtt{PK_{SA_k}}, \mathcal{H}))$: It takes system public parameter $\mathtt{PK_{SA_k}}$ and a role key hierarchy $\mathcal{H}$ as input. It outputs secret role parameter $\mathtt{RP_{SA_k}}$ and role public key $\mathbb{PK}^k_{ i}$, role secret $\mathtt{RS_{r^k_i}}$ of each role $r^k_{ i}\in \Psi_k$.
		\end{itemize}
		\item \emph{Key Generation}: In this phase, system administrator issues a pair of private, public keys, and user secrets for each registered users. In this phase, the role-manager also issues role-keys to each registered user according to the roles they hold. It comprises the following two algorithms.
		\begin{itemize}
			\item \textsc{PrivKeyGen} $((\mathtt{US_{ID^k_u}}, u^k_u, \mathtt{Pub_{ID^k_{u}}})\leftarrow (ID^k_{u}, \mathtt{MSK_{SA_k}}, \mathtt{PK_{SA_k}}))$: It takes system public parameter $\mathtt{PK_{SA_k}}$, master secret $\mathtt{MSK_{SA_k}}$ and unique identity $ID^k_{u}$ of a user as input. It outputs private key $u^k_u$, public key $\mathtt{Pub_{ID^k_{u}}}$ and user secret $\mathtt{US_{ID^k_u}}$.	 
			\item \textsc{RoleKeyGen} $(\mathtt{RK^k_{{m, u}}}\leftarrow \left(\mathtt{PK_{SA_k}}, \mathtt{RS_{r^k_m}}, g^{\delta_k}, \\\mathtt{US_{ID^k_u}}, r^k_{ m}\right))$: It takes system public parameter $\mathtt{PK_{SA_k}}$, role secret $\mathtt{RS_{r^k_m}}$, $g^{\delta_k}$, user secret $\mathtt{US_{ID^k_u}}$ and role $r^k_{ m}$ as input. It outputs a role-key $\mathtt{RK^k_{{m, u}}}$ of the role $r^k_{ m}$ for the user $ID_u^k$.
		\end{itemize}
		\item \emph{Encryption}: This phase is initiated by the data owners to encrypt data using role public keys according to RBAC access policy. It consists of \textsc{Enc} algorithm.
		\begin{itemize}
			\item \textsc{Enc} $(\mathbb{CT}\leftarrow (\mathtt{PK_{SA_k}, M,} \mathbb{PK}^k_{ i}))$: It takes system public parameter $\mathtt{PK_{SA_k}}$, a message $\mathtt{M}\in \mathcal{M}$ where $\mathcal{M}$ is the message space, and role public key $\mathbb{PK}^k_{ i}$ of a role $r^k_{ i}$ as input and outputs a ciphertext $\mathbb{CT}$. 
		\end{itemize}
		\item \emph{Decryption}: In this phase, a user of an organization accesses encrypted data hosted by the same organization by decrypting it using his/her role-keys and private key. It comprises \textsc{Dec} algorithm. 
		\begin{itemize}
			\item \textsc{Dec} $(\mathtt{M}\leftarrow (\mathbb{CT}, \mathtt{RK^k_{x, u}, u^k_u}))$: It takes role-key $\mathtt{RK^k_{x, u}}$, private key $u^k_u$ of a user $ID^k_{u}$ and ciphertext $\mathbb{CT}$ as input and outputs a message $\mathtt{M}$ if and only if $r^k_{ x}\in \mathbb{A}_{r^k_{ i}}$. 
		\end{itemize}
		\item \emph{User Revocation}: In this phase, system administrator revokes users from the system. It consists of \textsc{URevoke} algorithm.
		\begin{itemize}
			\item \textsc{URevoke ($\perp\leftarrow\mathtt{Pub_{ID^k_u}}$)}: It takes public key $\mathtt{Pub_{ID^k_u}}$ of a revoked user $ID^k_u$ as input. After invalidating the public key $\mathtt{Pub_{ID^k_u}}$ of the user $ID^k_u$, it outputs $\perp$.
		\end{itemize}
		\item \emph{Role Public Key Update}: In this phase, joint role public key is generated when an organization wants to share data with some other organizations. This phase is initiated by the system administrator of an organization. It consists of \textsc{RolePubKeyUpdate} algorithm.
		\begin{itemize}
			\item \textsc{RolePubKeyUpdate} $\left(\mathbb{PK}^{(k, k')}_{(i|| j)}\leftarrow \left(\mathbb{PK}^k_{ i}, \mathbb{PK}^{k'}_{ j}, \mathtt{PK_{SA_{k'}}}\right)\right)$: It takes role public key $\mathbb{PK}^k_{ i}$ of the role $r^k_{ i}$ which is maintained by the system administrator $\mathtt{SA_k}$ of the $k^{th}$ organization, another role public key $\mathbb{PK}^{k'}_{ j}$ of the role $r^{k'}_{ j}$ which is maintained by the system administrator $\mathtt{SA_{k'}}$ of the $k'^{th}$ organization and system public parameter $\mathtt{PK_{SA_{k'}}}$ as input. It outputs a joint role public key $\mathbb{PK}^{(k, k')}_{(i||j)}$ of the roles $r^k_{ i}$ and $r^{k'}_{ j}$.
		\end{itemize}
		\item \emph{Multi-Organization Encryption}: This phase encrypts the data using the joint role public key according to a RBAC access policy. It comprises \textsc{MultiEnc} algorithm.
		\begin{itemize}
			\item \textsc{MultiEnc} $(\mathbb{CT}\leftarrow (\mathtt{PK_{SA_k}}, \mathbb{PK}^{(k, k')}_{(i||j)}, \mathtt{M}))$: It takes public parameter $\mathtt{PK_{SA_k}}$ of $\mathtt{SA_k}$, joint role public key $\mathbb{PK}^{(k, k')}_{(i||j)}$ and the data (or message) $\mathtt{M}$ as input. It outputs a ciphertext $\mathbb{CT}$. 
		\end{itemize}
		\item \emph{System Administrator Agreement}: In this phase, an organization (i.e., system administrator) shares a secret key, termed as \emph{long term secret}, with other organizations. It consists of \textsc{LongKeyShare} algorithm.
		\begin{itemize}
			\item \textsc{LongKeyShare} $((\mathtt{ReKey_{k, k'}}, g^{(y_{k'}- \delta_{k'})\sigma_{k'}})\leftarrow (\mathtt{MSK_{SA_{k'}}}, \mathtt{PK_{SA_{k'}}}, \mathtt{MSK_{SA_{k}}}))$: It takes master secret $\mathtt{MSK_{SA_{k'}}}$ of the $k^{th}$ organization, public parameter $\mathtt{PK_{SA_{k'}}}$ and another master secret $\mathtt{MSK_{SA_{k}}}$ of the $k'^{th}$ organization as input. It outputs a \emph{long term secret key} $g^{(y_{k'}- \delta_{k'})\sigma_{k'}}$ and a proxy re-encryption key $\mathtt{ReKey_{k, k'}}$.
		\end{itemize} 
		\item \emph{Multi-Organization Decryption}: In this phase, a user of an organization accesses the encrypted data from other organizations by decrypting it using his/her role-key and private key. It consists of \textsc{MultiDec} algorithm. 
		\begin{itemize}
			\item \textsc{MultiDec} $(\mathtt{M} \leftarrow (\mathbb{CT}, \mathtt{ReKey_{k, k'}}, \eta_k, \sigma_{k'}, \mathtt{RK^{k'}_{x, u}, u^{k'}_u}))$: It takes a ciphertext $\mathbb{CT}$, proxy re-encryption key $\mathtt{ReKey_{k, k'}}$, secrets $\eta_k$ and $\sigma_{k'}$, role-key $\mathtt{RK^{k'}_{x, u}}$ and private key $u_{k', u}$ as input. It outputs plaintext data $\mathtt{M}$ if and only if $r^{k'}_{x}\in \mathbb{A}_{r^{k'}_{j}}$.
		\end{itemize}
	\end{enumerate}
	\begin{figure*}[t]
		\centering
		\scalebox{4}{\includegraphics[width=4.2cm, height=1.8cm]{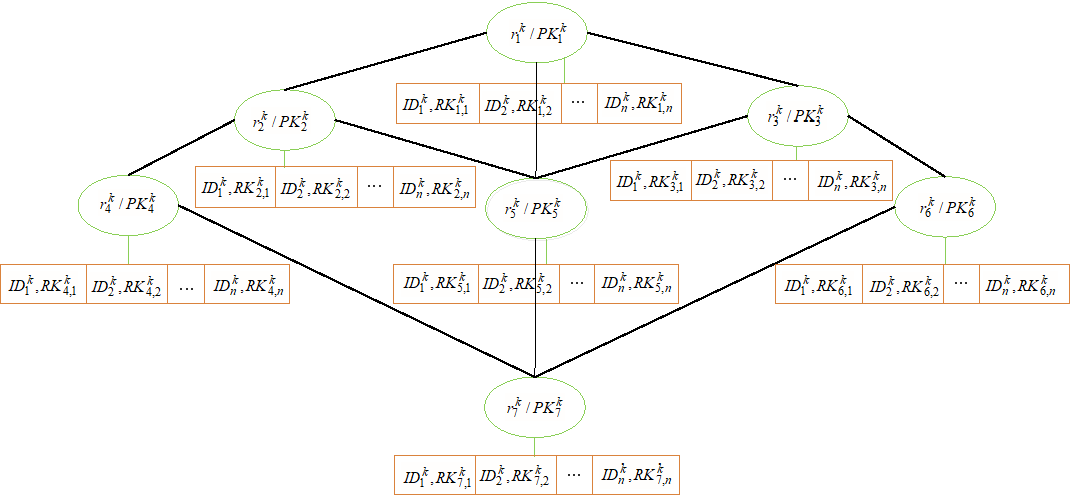}}
		\caption{Sample Role Key Hierarchy}
		\label{RKH}
	\end{figure*}
	\subsection{Security Assumptions}
	\label{RBE_assumption}
	In the proposed scheme, the following security assumptions are made.
	\begin{enumerate}
		\item [i)] Each system administrator and role-manager are fully trusted entities. 
		\item [ii)] Public and private clouds are honest-but-curious entities. They honestly perform assigned tasks, but they may try to gain all the possible knowledge about the outsourced data. 
		\item [iii)] Users are not honest. They may try to gain unauthorized access to the outsourced data by colluding with other users. 
		\item [iv)] As the public cloud is honest, we do not consider any active attacks from it by colluding with the revoked users as in \cite{Hur2011, Yu2010, Hur2013}, etc.
		\item [v)] All entities use secure channels while communicating with one another. The secure channels can be established using \emph{Secure Sockets Layer} (SSL).
		\item [vi)] There is an authentication mechanism which is used by role-managers and the public cloud to authenticate the users. 
	\end{enumerate}	
	
	\subsection{Security Model}
	\label{RBE_security_model}
	The security model of the proposed scheme is defined on \emph{Chosen Plaintext Attack} (CPA) security under \emph{selective ID Model}\footnote{In the \emph{Selective-ID security} model, the adversary must submit a challenged role and role hierarchy before the start of the security game. This is essential in our security proof to set up the role public key (please refer Section \ref{security_analsysis} for more details).} \cite{Sahai2005}. The CPA security can be illustrated using the following security game between a challenger $\mathcal{C}$ and an adversary $\mathcal{A}$. 
	
	\begin{itemize}
		\item {\bf \textsc{Initialization}} Adversary $\mathcal{A}$ sends arbitrary role hierarchies for each uncorrupted system authorities. It also sends two challenged roles $r^k_{i}$ and $r^{k'}_{i}$ of any two uncorrupted participating organizations/system authorities\footnote{Note that, the adversary can send more than two roles associated with any organization to the challenger. For simplicity, we consider two roles associated with two different organizations.} to the challenger $\mathcal{C}$.
		\item {\bf \textsc{Setup}} Challenger $\mathcal{C}$ runs \textsc{Init} to generate system public parameters $\mathtt{PK_{SA_k}}$ for each uncorrupted system authority $\mathtt{SA_k}$. It also generates role public key $\mathbb{PK}^k_{ i}$ for the role $r^k_{ i}$ of all the uncorrupted system authorities. It also generate re-encryption keys for each pair of uncorrupted system authorities. The challenger $\mathcal{C}$ sends the public parameters, i.e., $\mathtt{PK_{SA_k}}$, role public key $\mathbb{PK}^k_{ i}$ and the re-encryption keys to adversary $\mathcal{A}$.
		\item {\bf \textsc{Phase 1}} Adversary $\mathcal{A}$ submits two roles $r^k_{m}$ and $r^{k'}_m$ so that $r^k_{m}\notin \mathbb{A}_{r^k_{i}}$ (i.e., $r^k_{m}\in \bar{\mathbb{A}}_{r^k_{i}}$) and $r^{k'}_{m}\notin \mathbb{A}_{r^{k'}_{j}}$ (i.e., $r^{k'}_{m}\in \bar{\mathbb{A}}_{r^{k'}_{j}}$). It also sends an identity $ID^k_{u}$ to the challenger $\mathcal{C}$. Here the challenged identity can be associated with any of the uncorrupted system authorities. The challenger $\mathcal{C}$ initiates \textsc{PrivKeyGen} and \textsc{RoleKeyGen} algorithms to generate private key, public key and role key and sends them to the adversary $\mathcal{A}$. The adversary $\mathcal{A}$ sends queries for the secret keys to the challenger $\mathcal{C}$ by polynomially many times.
		\item {\bf \textsc{Challenge}}
		When the adversary $\mathcal{A}$ decides that {\bf \textsc{Phase 1}} is completed, he/she submits two equal length messages $\mathtt{K_0}$ and $\mathtt{K_1}$. The challenger $\mathcal{C}$ flips a random coin $b\in \{0, 1\}$ and encrypts the message $\mathtt{K_{b}}$ by initiating \textsc{MultiENC} algorithm.
		\par
		\item {\bf \textsc{Phase 2}} Same as {\bf \textsc{Phase 1}}.
		\item {\bf \textsc{Guess}} Adversary $\mathcal{A}$ outputs a guess $b'$ of $b$. The advantage of the adversary $\mathcal{A}$ to win this game is $Adv_{\mathcal{A}}= |Pr[b'= b]- \frac{1}{2}|$.
	\end{itemize}
	\par
	\begin{definition}
		The proposed scheme is secure against chosen plaintext attack if $Adv_{\mathcal{A}}$ is negligible for any polynomial time adversary $\mathcal{A}$.
	\end{definition}
	\begin{remark}
		Note that, the challenger sends re-encryption keys to the adversary $\mathcal{A}$ in the \textsc{Setup} phase. As such, the adversary $\mathcal{A}$ can re-encrypt the ciphertexts by itself. In addition, we do not consider the \textsc{Enc} oracle, as the adversary $\mathcal{A}$ can get the same response for the queries from the \textsc{MultiEnc} oracle. 
	\end{remark}
	
	\section{Proposed Scheme}
	\label{proposed_scheme}
	In this section, an overview of the proposed scheme is presented followed by its main constructions. In the construction, the Single-Organization Role-Based Encryption (SO-RBE) mechanism is presented first, followed by the Multi-Organization Role-Based Encryption (MO-RBE) mechanism.
	\subsection{Overview}
	Although the existing RBE schemes \cite{Zhou2013, Zhu2013} provide access control over encrypted data by enforcing RBAC access policy, these schemes cannot be applied in multi-organization cloud storage systems. Our main challenge is to construct a RBE scheme for both single and multi-organization cloud storage systems that support efficient decryption and user revocation. 
	\par    
	In the proposed model, the system administrator of an organization maintains a \emph{Role Key Hierarchy} (RKH) associated with a role hierarchy. In a RKH, each role is associated with a role public key and a group of users who hold that role. For each user in the group, a role-manager, who maintains that group, issues a unique role-key. The role-key is computed in such a way that it can decrypt any encrypted data computed using role public keys of any descendant roles. That is, if a user holds a role-key associated with the role $r_j$, the user can decrypt any ciphertext computed using role public key associated with the role $r_i$ if and only if $r_{j}$ belongs to the ancestor set of $r_i$ (i.e., $ r_j\in \mathbb{A}_{r_i}$). For better illustration, Figure \ref{RKH} shows a sample RKH of seven roles. Suppose, a data owner encrypts a message  $\mathtt{M_1}$ using the role public key $\mathtt{PK}^k_{7}$ then any user holding roles $r_1, r_2, r_3, r_4, r_5, r_6, r_7$ can decrypt it using their role-keys, as $\mathbb{A}_{r^k_{7}}= \{r_1, r_2, r_3, r_4, r_5, r_6, r_7\}$. Similarly, if the data owner encrypts another message $\mathtt{M_2}$ using role public key $\mathtt{PK}^k_{6}$ then the ciphertext can be decrypted only by the users who hold roles $r_1, r_3$ and $r_6$, as $\mathbb{A}_{r^k_6}= \{r_1, r_3, r_6\}$.
	\par 
	In order to share data with the users of other organizations, a joint role public key is computed by combining the role public keys of the participated organizations. The joint role public key is then used to encrypt data. The encrypted data is further re-encrypted in such a way that any user who holds a qualified role of a participated organization can decrypt. That is, the role held by the user should belong to an ancestor set of a role associated with the joint role public key. This implies that the same piece of information can be shared with the authorized users irrespective of which organization they belong to.
	
	\par 
	The outsourced decryption is achieved by enabling the public cloud to perform computationally expensive operations during decryption process. To achieve this, the users delegate their role-keys to the public cloud in such a way that the public cloud partially decrypts requested ciphertexts using the delegated role-keys without knowing the actual content of ciphertexts. 
	
	\subsection{Single-Organization RBE (SO-RBE)}
	\label{SO_RBE_construction}
	In this subsection, the proposed SO-RBE mechanism is presented. It consists of the following six phases.
	
	\subsubsection{System Initialization} In this phase, a system administrator, say $\mathtt{SA_k}$, generates system public parameter $\mathtt{PK_{SA_k}}$ and master secret $\mathtt{MSK_{SA_k}}$. The system public parameter $\mathtt{PK_{SA_k}}$ is kept in a public bulletin board while the master secret $\mathtt{MSK_{SA_k}}$ is kept secret. It consists of \textsc{Init} algorithm which is explained next.
	\begin{itemize}
		\item \textsc{Init} $((g^{\delta_k}, \mathtt{PK_{SA_k}}, \mathtt{MSK_{SA_k}})\leftarrow 1^{\Lambda_k})$: $\mathtt{SA_k}$ chooses two multiplicative cyclic groups $\mathbb{G}_1, \mathbb{G}_T$ of a large prime order $q$, a generator $g\in \mathbb{G}_1$, a collision resistant hash function $H_1: \{0, 1\}^w\rightarrow \mathbb{Z}_q^*$ and a bilinear map $\hat{e}: \mathbb{G}_1\times \mathbb{G}_1\rightarrow \mathbb{G}_T$. It also chooses random numbers $y_k, \delta_k, \sigma_k, \eta_k\in \mathbb{Z}_q^*$. It then computes $Y_k= \hat{e}(g, g)^{y_k}, V_k= \hat{e}(g, g)^{\delta_k}$ and $h_k= g^{\eta_k}$. The system public parameter is $\mathtt{PK_{SA_k}}$, where $\mathtt{PK_{SA_k}}= \left<q, \mathbb{G}_1, \mathbb{G}_T, g, \hat{e}, H_1, Y_k, V_k, h_k\right>$ and the master secret is $\mathtt{MSK_{SA_k}}$, where $\mathtt{MSK_{SA_k}}= \left<y_k, \delta_k, \sigma_k, \eta_k\right>$. $\mathtt{SA_k}$ keeps $\sigma_k$ and $\eta_k$ in its private cloud. $\mathtt{SA_k}$ also sends $g^{\delta_k}$ to each role-manager of the organization.
	\end{itemize}

	\subsubsection{Manage Role} In this phase, a system administrator, say $\mathtt{SA_k}$, generates secret role parameter $\mathtt{RP_{SA_k}}$. $\mathtt{SA_k}$ also generates role public key $\mathbb{PK}^k_{i}$ and role secret $\mathtt{RS_{r^k_i}}$ for each role $r^k_i$ in the organization. $\mathtt{SA_k}$ keeps the role public keys $\{\mathbb{PK}^k_{i}\}_{r^k_{ i}\in \Psi_k}$ in its public bulletin board and keeps secret role parameter $\mathtt{RP_{SA_k}}$ in a secure place. It also sends role secrets to the respective role-manager. This phase comprises \textsc{RoleParaGen} algorithm which is described next.
	\begin{itemize}
		\item \textsc{RoleParaGen} $((\mathtt{RP_{SA_k}}, \{\mathbb{PK}^k_{i}\}_{r^k_{ i}\in \Psi_k}, \mathtt{(RS_{r^k_i})_{\forall r^k_{ i}\in \Psi_k}})\leftarrow (\mathtt{PK_{SA_k}}, \mathcal{H}))$: $\mathtt{SA_k}$ chooses random numbers $[t_{k, i}]_{r^k_{ i}\in \Psi_k}\in \mathbb{Z}_q^*$. For each role $r^k_{ i}$ in $\Psi_k$, it computes role public keys $\mathbb{PK}^k_{i}= \left<\mathtt{PK^k_{{i}}}, \{\mathtt{AR^{k}_l} | r^{k}_l \in \mathbb{A}_{r^{k}_i}\}, r^k_{i}\right>$ and role secret $\mathtt{RS_{r^k_i}}$, where  
		\begin{align*}
		\mathtt{PK^k_{{i}}}= & g^{\sum_{r^k_j\in \bar{\mathbb{A}}_{r^k_i}}t_{k, j}} \\
		\mathtt{AR^k_{ l}}= &g^{t_{k, l}}\\
		\mathtt{RS_{r^k_i}}= & \frac{1}{\sum_{ r^k_{j}\in \bar{\mathbb{A}}_{r^k_{i}}}t_{k, j}}
		\end{align*}
		The secret role parameter is $\mathtt{RP_{SA_k}}$, where $\mathtt{RP_{SA_k}}= \left<[t_{k, i}]_{r^k_{ i}\in \Psi_k}\right>$.
	\end{itemize}
\begin{remark}
	With this scheme, it is assumed that the root role in the role hierarchy has more than two children roles to ensure that the role secrets are not disclosed to any other role manager (as with only two children roles each role-manager will know the other's secret role parameter).
\end{remark}
	\subsubsection{Key Generation} In this phase, a system administrator issues a pair of private, public keys and user secrets for each registered users. The role-manager also issues role-keys to the registered user. The issued private and role-keys are sent to the user using secure-channels while the public key is kept in the public bulletin board. On the other hand, system administrator keeps the user secrets in its private cloud so that the role-manager can access the stored user secrets from the private cloud during role-key generation process. This phase consists of two algorithms, namely \textsc{PrivKeyGen} and \textsc{RoleKeyGen} which are given next.
	\begin{itemize}
		\item \textsc{PrivKeyGen} $((\mathtt{US_{ID^k_u}}, u^k_u, \mathtt{Pub_{ID^k_{u}}})\leftarrow (ID^k_{u}, \mathtt{MSK_{SA_k}}, \mathtt{PK_{SA_k}}))$: Suppose a user, say $ID^k_u$, joins the system for the first time. System administrator, say $\mathtt{SA_k}$, chooses a random number $u^k_u\in \mathbb{Z}_q^*$ as a private key and computes the public key $\mathtt{Pub_{ID^k_{u}}}$ and user secret $\mathtt{US_{ID^k_u}}$, where 
		\begin{align*}
		\mathtt{Pub_{ID^k_{u}}}= &g^{\frac{\left[u^k_u+ H_1(ID^k_{u})\right]\delta_k}{\eta_k}}\\
		\mathtt{US_{ID^k_u}}= & (g^{y_k})^{u^k_u}= g^{y_k\cdot u^k_u}
		\end{align*}	
		\item \textsc{RoleKeyGen} $\left(\mathtt{RK^k_{{m, u}}}\leftarrow \left(\mathtt{PK_{SA_k}}, \mathtt{RS_{r^k_m}}, g^{\delta_k}, \mathtt{US_{ID^k_u}}, r^k_{ m}\right)\right)$: When a role-manager $\mathtt{RM_{r^k_m}}$ assigns a role to a user, this algorithm is initiated. Let a user $ID^k_{u}$ holds a role $r^k_{ m}$. The role-manager $\mathtt{RM_{r^k_m}}$ first authenticate the user $ID^k_{u}$ and gets his/her user secret $\mathtt{US_{ID^k_u}}$ from the private cloud. It then issues a role-key $\mathtt{RK^k_{{m, u}}}$, where
		\begin{align*}
		\mathtt{RK^k_{{m, u}}}&= \left(\mathtt{US_{ID^k_u}}\cdot (g^{\delta_k})^{H_1\left(ID^k_{u}\right)}\right)^{\mathtt{RS_{r^k_m}}}\\
		&= \left(g^{y_k\cdot u^k_u}\cdot (g^{\delta_k})^{H_1\left(ID^k_{u}\right)}\right)^{\frac{1}{\sum_{r^k_{j}\in \bar{\mathbb{A}}_{r^k_{ m}}}t_{k, j}}}\\
		&= g^{\frac{\left[y_k\cdot u^k_u+ H_1\left(ID^k_{u}\right)\cdot \delta_k\right]}{\sum_{r^k_{j}\in \bar{\mathbb{A}}_{r^k_{ m}}}t_{k, j}}}
		\end{align*}
		It is to be noted that $g^{\delta_k}$ and $\mathtt{RS_{r^k_m}}$ are known to the role-manager $\mathtt{RM_{r^k_m}}$.
	\end{itemize}

	\subsubsection{Encryption} 
	\label{RBE_encryption}
	In this phase, a data owner encrypts data before outsourcing it to the public cloud. The data owner first encrypts a random secret key using a RBAC access policy and then encrypts the actual data using the secret key. For the secret key encryption part, \textsc{Enc} algorithm is used, which is defined next. While, for the actual data encryption part, any secure symmetric key encryption algorithm, like Advanced Encryption Standards (AES), can be used. Finally, the data owner combines both the encrypted files and outsources it to the public cloud. 
	\begin{itemize}
		\item \textsc{Enc} $(\mathbb{CT}\leftarrow  (\mathtt{PK_{SA_k}, M,} \mathbb{PK}^k_{i}))$: Let a data owner wants to encrypt message $\mathtt{M}\in \mathcal{M}$, where $\mathcal{M}$ is the message space, using role $r^k_{ i}$. First, the message $\mathtt{M}$ is encrypted using a random secret key, say $\mathtt{K}\in \mathbb{G}_T$, and generates a ciphertext $\mathtt{Enc_{K}(M)}$. Afterwards, the secret key $\mathtt{K}$ is encrypted using the role public key $\mathbb{PK}^k_{i}$ of the role $r^k_{ i}$. The encryption procedure is explained below.
		\begin{itemize}
			\item data owner chooses a random number $d\in \mathbb{Z}_q^*$ and computes $C_1, C_2, \{C_{3l}|  \forall r^k_{ l}\in \mathbb{A}_{r^k_i}\}$ and $C_{k, i}$, where 		
			\begin{align*}
			C_1= &\mathtt{K} \left(\frac{Y_k}{V_k}\right)^d\\= &\mathtt{K}\cdot \frac{\hat{e}(g, g)^{y_k\cdot d}}{\hat{e}(g, g)^{\delta_k\cdot d}}\\= &\mathtt{K}\cdot \hat{e}(g, g)^{y_k\cdot d-\delta_k\cdot d}\\
			C_2= &\left(h_k\right)^d= g^{\eta_k\cdot d}\\
			C_{3l}= & (\mathtt{AR^k_l})^d= g^{d\cdot t_{k, l}}\\
			C_{k, i}= &(\mathtt{PK^k_{{i}}})^d= g^{d\sum_{ r^k_{j}\in \bar{\mathbb{A}}_{r^k_{i}}}t_{k, j}}
			\end{align*}
			\item finally, the data owner generates a ciphertext $\mathbb{CT}$, where \begin{align*}
			\mathbb{CT}= &\big<\mathtt{Enc_{K}(M)}, C_1, C_2, \{C_{3l}|  \forall r^k_{ l}\in \mathbb{A}_{r^k_i}\}, C_{k, i}, \\&r^k_{ i}\big>\end{align*}
		\end{itemize}
	\end{itemize}
	\subsubsection{Decryption} 
	\label{Decryption_SR-RBE}	
	In this phase, a user accesses encrypted data by decrypting it using his/her private and role-keys. To take advantage of the outsourced decryption, the user first sends a transformed role-key to the public cloud, and the public cloud partially decrypts the requested ciphertexts using the transformed role-key and the public key of the user. Afterwards, the user decrypts the partially decrypted ciphertexts using his/her private key. This phase comprises \textsc{DEC} algorithm, which is described next.
	\begin{itemize}
		\item \textsc{Dec} $(\mathtt{M}\leftarrow (\mathbb{CT}, \mathtt{RK^k_{x, u}, u^k_u}))$: Suppose a user $ID^k_{u}$, who holds a role $r^k_{ x}$, wants to decrypt a ciphertext $\mathbb{CT}= \left<\mathtt{Enc_{K}(M)}, C_1, C_2, \{C_{3l}|  \forall r^k_{ l}\in \mathbb{A}_{r^k_i}\}, C_{k, i}, r^k_{ i}\right>$, where $r^k_{ x}\in \mathbb{A}_{r^k_i}$. The user $ID^k_{u}$ sends a data access request to the public cloud along with a transformed role-key $\mathtt{TRK^k_{x, u}}$ and $r^k_{ x}$. The user $ID^k_{u}$ chooses a random number $v\in \mathbb{Z}_q^*$ and computes $\mathtt{TRK^k_{x, u}}$, where
		\begin{align*}
		\mathtt{TRK^k_{x, u}}& = (\mathtt{RK^k_{{x, u}}})^{v}= g^{\frac{v[u^k_u\cdot y_k+ H_1(ID^k_{u})\cdot \delta_k]}{\sum_{ r^k_{j}\in \bar{\mathbb{A}}_{r^k_x}}t_{k, j}}}
		\end{align*}
		Afterwards, the user keeps the random number $v$ in a secure database for future use.
		\par 
		The public cloud initiates outsourced decryption process once it received the transformed role-key $\mathtt{TRK^k_{{ x, u}}}$ from the user $ID^k_u$. In the outsourced decryption process, public cloud generates two ciphertext components, namely $P$ and $Q$. Let $T= [u^k_u\cdot y_k+ H_1(ID^k_{u})\cdot \delta_k]$ and $\Gamma(r^k_{ x}, r^k_{ i})$ denotes $\left(\mathbb{D}_{r^k_x}\setminus \mathbb{D}_{r^k_i}\right)$. Public cloud computes $P$ and $Q$ as follows:
		\begin{align*}
		P&= \hat{e}\left(C_{k, i}\prod_{r^k_{ l}\in \Gamma (r^k_{ x}, r^k_{ i})}C_{3l}, \mathtt{TRK^k_{{ x, u}}}\right)\\
		&= \hat{e}\left(g^{d\sum_{ r^k_j\in \bar{\mathbb{A}}_{r^k_i}}t_{k, j}}\prod_{r^k_{ l}\in \Gamma (r^k_{ x}, r^k_{ i})}g^{d\cdot t_{k, l}}, g^{\frac{v\cdot T}{\sum_{ r^k_l\in \bar{\mathbb{A}}_{r^k_x}}t_{k, l}}}\right)\\
		&= \hat{e}\left(g^{d\sum_{r^k_l\in \bar{\mathbb{A}}_{r^k_x}}t_{k, l}}, g^{\frac{v\cdot T}{\sum_{ r^k_l\in \bar{\mathbb{A}}_{r^k_x}}t_{k, l}}}\right)\\
		&= \hat{e}(g, g)^{d\cdot v\cdot T}= \hat{e}(g, g)^{d\cdot v\cdot [u^k_u\cdot y_k+ H_1(ID^k_{u})\cdot \delta_k]}\\
		Q&= \hat{e}\left(C_2, \mathtt{Pub_{ID^k_{u}}}\right)\\
		&= \hat{e}\left(g^{\eta_k\cdot d}, g^{\frac{[u^k_u+ H_1(ID^k_{u})]\delta_k}{\eta_k}}\right)= \hat{e}(g, g)^{d[u^k_u+ H_1(ID^k_{u}) ]\delta_k}	 
		\end{align*} 
		
		Later on, public cloud sends the newly computed ciphertext components $\mathtt{Enc_{K}(M)}$, $C_1$, $P$ and $Q$ to the user. The user $ID^k_{u}$ gets $\mathtt{K}$ using his/her private key $u^k_u$ and random secret $v$ as follows:
		\begin{align*}
		\mathtt{K}= &\frac{C_1}{X},\text{ where}\\
		X= &\left(\frac{P^{\frac{1}{v}}}{Q}\right)^{\frac{1}{u^k_u}}\\
		= &\left(\frac{\hat{e}(g, g)^{\frac{d[u^k_u\cdot y_k+ H_1(ID^k_{u})\cdot \delta_k]v}{v}}}{\hat{e}(g, g)^{d[u^k_u+ H_1(ID^k_{u})]\delta_k}}\right)^{\frac{1}{u^k_u}}\\
		= &\left(\frac{\hat{e}(g, g)^{d[u^k_u\cdot y_k+ H_1(ID^k_{u})\cdot \delta_k]}}{\hat{e}(g, g)^{d[u^k_u+ H_1(ID^k_{u})]\delta_k}}\right)^{\frac{1}{u^k_u}}\\
		= &\left(\hat{e}(g. g)^{d\cdot u^k_u\cdot y_k- d\cdot u^k_u\cdot \delta_k}\right)^{\frac{1}{u^k_u}}= \hat{e}(g. g)^{d\cdot y_k- d\cdot \delta_k}	
		\end{align*}
		Afterwards, the user $ID^k_{u}$ decrypts $\mathtt{Enc_{K}(M)}$ using $\mathtt{K}$ and gets the actual message $\mathtt{M}$. Finally, the user $ID^k_{u}$ deletes the random secret $v$ from his/her database.
	\end{itemize}
	\subsubsection{User Revocation}
	\label{user_revocation}
	System administrator revokes users in this phase. It consists of \textsc{URevoke} algorithm which is described below.
	\begin{itemize}
		\item \textsc{URevoke ($\perp\leftarrow\mathtt{Pub_{ID^k_u}}$)}: Suppose, $\mathtt{SA_k}$ wants to revoke a user, say $ID^k_{u}$. $\mathtt{SA_k}$ simply removes the public key $\mathtt{Pub_{ID^k_{u}}}$ of the user $ID^k_{u}$ from its public bulletin board so that public cloud can no longer use it for the outsourced decryption process. This, in turn, prevents the revoked user from accessing data.
	\end{itemize}
	
	\subsection{Extension to Multi-Organization RBE (MO-RBE)}
	\label{multi-RBE}
	The SO-RBE mechanism can only be used to achieve access control over encrypted data in a single-organization cloud storage system. In this subsection, our MO-RBE mechanism is presented to achieve access control over encrypted data in the multi-organization cloud storage system. The MO-RBE is an extension of the SO-RBE mechanism which requires the following additional phases.

	\subsubsection{Role Public Key Update} In this phase, the system administrator of an organization computes joint role public key when it wants to share data with the users of other organizations. This phase comprises \textsc{RolePubKeyUpdate} algorithm which is described next.
	\begin{itemize}
		\item \textsc{RolePubKeyUpdate} ($\mathbb{PK}^{(k, k')}_{({i}||{j})}\leftarrow (\mathbb{PK}^k_{i}, \mathbb{PK}^{k'}_{j}, \mathtt{PK_{SA_{k'}}})$): Suppose $k^{th}$ organization wants to share data with the users of $k^{'th}$ organization. Let the system administrator $\mathtt{SA_k}$ wants to share data with its users who hold access privilege for the role $r^{k}_i$ and also with the users of $\mathtt{SA_{k'}}$ holding access privilege for the role $r^{k'}_j$. The system administrator $\mathtt{SA_k}$ computes a joint role public key $\mathbb{PK}^{(k, k')}_{({i}||{j})}$ for the roles $r^k_i$ and $r^{k'}_j$, where  
		\begin{align*}
		\mathbb{PK}^{(k, k')}_{({i}||{j})}= &\big<\mathtt{PK^k_{{i}}}, \mathtt{PK^{k'}_{{j}}},  \{\mathtt{AR^{k}_l} | r^{k}_l \in \mathbb{A}_{r^{k}_i}\},  \{\mathtt{AR^{k'}_l} | r^{k'}_l \in \mathbb{A}_{r^{k'}_j}\},\\& r^k_{ i}, r^{k'}_{ j}\big>
		\end{align*}
		System administrator $\mathtt{SA_k}$ knows $\mathtt{PK^{k'}_{j}}$ and $\{\mathtt{AR^{k'}_l} | r^{k'}_l \in \mathbb{A}_{r^{k'}_j}\}$ from the role public key $\mathbb{PK}^{k'}_{j}$ which is available in the public bulletin board of the $k^{'th}$ organization.
	\end{itemize}
	
	\subsubsection{Multi-Organization Encryption}
	\label{multi-encryption}
	In this phase, a data owner of an organization encrypts data using the joint role public keys so that the encrypted data can be accessed by the authorized users of other organizations as well as his/her own organization by decrypting it. This phase consists of \textsc{MultiEnc} algorithm.
	
	\begin{itemize}
		\item \textsc{MultiEnc} $(\mathbb{CT}\leftarrow (\mathtt{PK_{SA_k}}, \mathbb{PK}^{(k, k')}_{({i}||{j})}, \mathtt{M}))$: Let the data owner wants to share a message $\mathtt{M}$ with the users of his/her organization who hold access privilege for the role $r^k_{i}$ and the users of another organization, say $k'$ organization, who hold access privilege for the role $r^{k'}_j$. Like the \emph{Encryption} phase described in Section \ref{RBE_encryption}, the data owner first encrypts data $\mathtt{M}$ using a random secret key $\mathtt{K}\in \mathbb{G}_T$ and the secret key $\mathtt{K}$ is encrypted using the joint role public key $\mathbb{PK}^{(k, k')}_{({i}||{j})}$ of the roles $r^k_{ i}$ and $r^{k'}_{ j}$. Finally, the data owner outsources both the encrypted data to the cloud storage servers. This encryption procedure is done as follows:
		
		\begin{itemize}
			\item the data owner chooses a random number $d\in \mathbb{Z}_q^*$ and computes $C_1, C_2, C'_2$, where 
			\begin{align*}
			C_1= &\mathtt{K} (\frac{Y_k}{V_k})^d= \mathtt{K}\cdot \frac{\hat{e}(g, g)^{y_k\cdot d}}{\hat{e}(g, g)^{\delta_k\cdot d}}= \mathtt{K}\cdot \hat{e}(g, g)^{y_k\cdot d-\delta_k\cdot d}\\
			C_{2}= &(h_k)^d= g^{\eta_{k}\cdot d}\\
			C'_{2}= &(h_{k'})^d= g^{\eta_{k'}\cdot d}
			\end{align*}
			\item afterwards the data owner computes $C_{k, i}, \{C_{3l}| \forall r^k_{ l}\in \mathbb{A}_{r^k_{ i}}\}, C_{k', j}$ and $\{C'_{3l}| \forall r^{k'}_{ l}\in \mathbb{A}_{r^{k'}_{ j}}\}$, where 
			\begin{flalign*}
			C_{k, i}= &(\mathtt{PK^k_{{i}}})^d= g^{d\sum_{ r^k_{ z}\in \bar{\mathbb{A}}_{r^k_{i}}}t_{k, z}}\\
			C_{3l}= & (\mathtt{AR^k_l})^d= g^{d\cdot t_{k, l}}\\
			C_{k', j}= &(\mathtt{PK^{k'}_{{j}}})^d= g^{d\sum_{r^{k'}_{ z}\in \bar{\mathbb{A}}_{r^{k'}_{ j}}}t_{k', z}}\\
			C'_{3l}= & (\mathtt{AR^{k'}_{l}})^d= g^{d\cdot t_{k', l}}
			\end{flalign*}
			\item finally, the data owner generates a ciphertext $\mathbb{CT}$, where 
			\begin{align*}
			\mathbb{CT}= &\big<\mathtt{Enc_{K}(M)}, C_1, C_{2}, C'_{2}, C_{k, i}, \{C_{3l}| \forall r^k_{ l}\in \mathbb{A}_{r^k_{ i}}\},\\& C_{k', j}, \{C'_{3l}| \forall r^{k'}_{ l}\in \mathbb{A}_{r^{k'}_{ j}}\}, r^k_{ i}, r^{k'}_{ j}\big>
			\end{align*}
		\end{itemize}
	\end{itemize}
	
	\subsubsection{System Administrator Agreement}
	\label{agreement}
	In this phase, system administrator of an organization share a secret key, termed as \emph{long term secret}, with system administrators of other organizations so that its users can access data stored by the other organizations. This phase comprises \textsc{LongKeyShare} algorithm which is described next.
	\begin{itemize}
		\item \textsc{LongKeyShare} $((\mathtt{ReKey_{k, k'}}, g^{(y_{k'}- \delta_{k'})\sigma_{k'}})\leftarrow (\mathtt{MSK_{SA_{k'}}}, \mathtt{PK_{SA_{k'}}}, \mathtt{MSK_{SA_{k}}}))$: Let the $k^{'th}$ organization wants to enable its users to access data hosted by $k^{th}$ organization. $\mathtt{SA}_{k'}$ computes the \emph{long term secret} $g^{(y_{k'}- \delta_{k'})\sigma_k'}$ using its master secret $\mathtt{MSK_{SA_{k'}}}$ and sends it to $\mathtt{SA_{k}}$. Afterwards, $\mathtt{SA_{k}}$ computes a proxy re-encryption key $\mathtt{ReKey_{k, k'}}$ and keeps it in its private cloud, where
		\begin{align*}
		\mathtt{ReKey_{k, k'}}&= \frac{g^{(y_{k'}- \delta_{k'})\sigma_{k'}}}{g^{(y_{k}- \delta_k)}}= g^{(y_{k'}- \delta_{k'})\sigma_{k'}- (y_k- \delta_k)}
		\end{align*}
	\end{itemize}
	It is to be noted that this phase needs to be performed only once between two organizations when they agreed to share data.
	
	\subsubsection{Decryption}
	\label{multi-decryption}
	In this phase, a user of an organization access encrypted data stored by another organization. This phase consists of \textsc{MultiDec} algorithm, which is defined as follows.
	\begin{itemize}
		\item \textsc{MultiDec} $(\mathtt{M} \leftarrow (\mathbb{CT},  \mathtt{ReKey_{k, k'}}, \eta_k, \sigma_{k'},\mathtt{RK^{k'}_{x, u}, u^{k'}_u}))$: Suppose a user $ID^{k'}_{u}$ having role $r^{k'}_{ x}$ wants to access encrypted data $\mathbb{CT}= \big<\mathtt{Enc_{K}(M)}, C_1, C_{2}, C'_{2}, C_{k, i}, \{C_{3l}| \forall r^k_{ l}\in \mathbb{A}_{r^k_{ i}}\}, C_{k', j}, \\\{C'_{3l}| \forall r^{k'}_{ l}\in \mathbb{A}_{r^{k'}_{ j}}\}, r^k_{ i}, r^{k'}_{ j}\big>$, where $r^{k'}_{x}\in \mathbb{A}_{r^{k'}_{j}}$. Similar with the \emph{Decrypt} phase in Section \ref{Decryption_SR-RBE}, the user $ID^{k'}_{u}$ sends a data access request to the public cloud along with a transformed role key $\mathtt{TRK^{k'}_{x, u}}$ and $r^{k'}_{ x}$, where
		\begin{align*}
		\mathtt{TRK^{k'}_{x, u}}& = (\mathtt{RK^{k'}_{{x, u}}})^{v}= g^{\frac{v[u^{k'}_u\cdot y_{k'}+ H_1(ID^{k'}_{u})\cdot \delta_{k'}]}{\sum_{r^{k'}_{z}\in \bar{\mathbb{A}}_{r^{k'}_{x}}}t_{k', z}}}
		\end{align*}
		
		where $v\in \mathbb{Z}_q^*$ is a random number which is kept for future use.
		\par 
		After receiving the data access request, public cloud forwards the transformed role key $\mathtt{TRK^{k'}_{x, u}}$ to the private cloud of $\mathtt{SA_{k'}}$, and it also sends the ciphertext components $C_1$ and $C_{2}$ to the private cloud of $\mathtt{SA_k}$. 
		\par 	
		The private cloud of $\mathtt{SA_{k'}}$ computes a temporary decryption key $\mathtt{TDK^{k'}_{x, u}}$ and $\left(\mathtt{Pub_{ID^{k'}_{u}}}\right)^{\sigma_{k'}}$ using $\mathtt{TRK^{k'}_{x, u}}$ and $\sigma_{k'}$, and sends them to the public cloud, where $\mathtt{Pub_{ID^{k'}_{u}}}$ is the public key of the user $ID^{k'}_u$ and 
		\begin{align*}
		\mathtt{TDK^{k'}_{x, u}}=& (\mathtt{TRK^{k'}_{x, u}})^{\sigma_{k'}}= g^{\frac{v\left[u^{k'}_u\cdot y_{k'}+ H_1(ID^{k'}_{u})\cdot \delta_{k'}\right]\sigma_{k'}}{\sum_{r^{k'}_{z}\in \bar{\mathbb{A}}_{r^{k'}_{x}}}t_{k', z}}}
		\end{align*}
		It is to be noted that $\sigma_{k'}$ is known to the private cloud of $\mathtt{SA_{k'}}$. On the other hand, the private cloud of $\mathtt{SA_k}$ translates (i.e., re-encrypts) the ciphertext component from $C_1$ to $C_1'$ using the proxy re-encryption key $\mathtt{ReKey}_{k, k'}$ and $C_{2}$ as follows:
		\begin{align*}
		C_1'&= C_1\cdot X\\
		&=  \mathtt{K}\cdot \hat{e}(g, g)^{y_k\cdot d-\delta_k\cdot d}\cdot X\\
		&= \mathtt{K}\cdot \hat{e}(g, g)^{(y_{k'}-\delta_{k'})\sigma_{k'}\cdot d}, \text{ where}\\
		X&= \hat{e}\left(\mathtt{ReKey_{k, k'}}, (C_{2})^{\frac{1}{\eta_k}}\right)\\
		&= \hat{e}\left(g^{(y_{k'}- \delta_{k'})\sigma_{k'}- (y_k- \delta_k)}, \left(g^{\eta_k\cdot d}\right)^{\frac{1}{\eta_k}}\right)\\
		&= \hat{e}\left(g, g\right)^{(y_{k'}- \delta_{k'})\sigma_{k'}\cdot d- (y_k- \delta_k) d}\\
		&= \hat{e}\left(g, g\right)^{(y_{k'}- \delta_{k'})\sigma_{k'}\cdot d- (y_k\cdot d- \delta_k\cdot d) }
		\end{align*}
		Afterwards, the private cloud $\mathtt{SA_k}$ sends the translated ciphertext component $C_1'$ to the public cloud.
		\par 
		
		After receiving $\mathtt{TDK^{k'}_{x, u}}$, $\left(\mathtt{Pub_{ID^{k'}_{u}}}\right)^{\sigma_{k'}}$ and $C_1'$, the public cloud partially decrypt the ciphertext $\mathbb{CT}'= \{\mathbb{CT}\cup C_1'\}$ using $\mathtt{TDK^{k'}_{x, u}}$ and $\left(\mathtt{Pub_{ID^{k'}_{u}}}\right)^{\sigma_{k'}}$. The public cloud computes $P$ and $Q$ as follows: let $T= [u^{k'}_u\cdot y_{k'}+ H_1(ID^{k'}_{u})\cdot \delta_{k'}]$ and $\Gamma(r^{k'}_{ x}, r^{k'}_{ j})$ denotes $\left(\mathbb{D}_{r^{k'}_x}\setminus \mathbb{D}_{r^{k'}_j}\right)$
		
		\begin{scriptsize}
		\begin{align*}
		P&= \hat{e}\left(C_{k', j}\cdot \prod_{r^{k'}_{ l}\in \Gamma (r^{k'}_{ x}, r^{k'}_{ j})}C'_{3l}, \mathtt{TDK^{k'}_{x, u}}\right)\\
		&= \hat{e}\left(g^{d\sum_{r^{k'}_{ z}\in \bar{\mathbb{A}}_{r^{k'}_{ j}}}t_{k', z}}\cdot \prod_{r^{k'}_{ l}\in \Gamma (r^{k'}_{ x}, r^{k'}_{ j})}g^{d\cdot t_{k', l}}, g^{\frac{v\cdot T\cdot \sigma_{k'}}{\sum_{r^{k'}_{ l}\in \bar{\mathbb{A}}_{r^{k'}_{x}}}t_{k', l}}}\right)\\
		&= \hat{e}\left(g^{d\sum_{r^{k'}_{ l}\in \bar{\mathbb{A}}_{r^{k'}_{ x}}}t_{k', l}}, g^{\frac{v\cdot T\cdot \sigma_{k'}}{\sum_{r^{k'}_{ l}\in \bar{\mathbb{A}}_{r^{k'}_{ x}}}t_{k', l}}}\right)\\
		&= \hat{e}\left(g, g\right)^{d\cdot T\cdot \sigma_{k'}\cdot v}\\
		&= \hat{e}\left(g, g\right)^{\left[u^{k'}_u\cdot y_{k'}+ H_1(ID^{k'}_{u})\cdot \delta_{k'}\right]\sigma_{k'}\cdot v} \text{, (replacing value of $T$)}\\
		Q&= \hat{e}\left(C'_{2}, \left(\mathtt{Pub_{ID^{k'}_{u}}}\right)^{\sigma_{k'}}\right)\\
		&= \hat{e}\left(g^{\eta'_k\cdot d}, g^{\frac{[u^{k'}_u+ H_1(ID^{k'}_{u})]\delta_{k'}\cdot \sigma_{k'}}{\eta_k'}}\right)\\
		&= \hat{e}\left(g, g\right)^{d\left[u^{k'}_u+ H_1(ID^{k'}_{u})\right]\delta_{k'}\cdot \sigma_{k'}}	 
		\end{align*} \end{scriptsize}
	
		The ciphertext components $\mathtt{Enc_{K}(M)}$ and $C_1'$ are now sent to the user $ID^{k'}_{u}$ along with $P$ and $Q$. Finally, user $ID^{k'}_{u}$ gets $\mathtt{K}$ using his/her private key $u^{k'}_u$ and random secret $v$ as follows:
		
		\begin{align*}
		\mathtt{K}= &\frac{C'_1}{R},\text{ where}\\
		R= &\left(\frac{P^{\sfrac{1}{v}}}{Q}\right)^{(\sfrac{1}{u^{k'}_u})}\\
		=&  \left(\frac{\left(\hat{e}(g, g)^{d\cdot\left[u^{k'}_u\cdot y_{k'}+ H_1(ID^{k'}_{u})\cdot \delta_{k'}\right]\sigma_{k'}\cdot v }\right)^{\sfrac{1}{v}}}{\hat{e}(g, g)^{d\left[u^{k'}_u+ H_1(ID^{k'}_{u})\right]\delta_{k'}\cdot \sigma_{k'}}}\right)^{\sfrac{1}{u^{k'}_u}}\\
		= & \left(\frac{\hat{e}(g, g)^{d\cdot\left[u^{k'}_u\cdot y_{k'}+ H_1(ID^{k'}_{u})\cdot \delta_{k'}\right]\sigma_{k'}}}{\hat{e}(g, g)^{\sigma_{k'}\cdot d\left[u^{k'}_u+ H_1(ID^{k'}_{u})\right]\delta_{k'}}}\right)^{\sfrac{1}{u^{k'}_u}}\\
		= & \left(\frac{\hat{e}\left(g, g\right)^{\left[\sigma_{k'}\cdot d\cdot u^{k'}_u\cdot y_{k'}+ \sigma_{k'}\cdot d\cdot H_1(ID^{k'}_{u})\cdot \delta_{k'}\right]}}{\hat{e}(g, g)^{\left[\sigma_{k'}\cdot d\cdot u^{k'}_u \cdot \delta_{k'}+ \sigma_{k'}\cdot d\cdot H_1(ID^{k'}_{u})\cdot \delta_{k'}\right]}}\right)^{\sfrac{1}{u^{k'}_u}}\\
		= & \left(\hat{e}(g, g)^{(y_{k'}- \delta_{k'})\sigma_{k'}\cdot d\cdot u^{k'}_u}\right)^{\sfrac{1}{u^{k'}_u}}\\
		= & \hat{e}\left(g, g\right)^{(y_{k'}- \delta_{k'})\sigma_{k'}\cdot d}
		\end{align*}
		
		Later on, the user $ID^{k'}_{u}$ decrypts $\mathtt{Enc_{K}(M)}$ using $\mathtt{K}$ and gets $\mathtt{M}$. Finally, the user $ID^{k'}_{u}$ removes random secret $v$ from its database.
	\end{itemize}
	\begin{remark}
		In the \emph{Decryption} phase (Section \ref{multi-decryption}) of MO-RBE, any user from the $k^{th}$ organization (i.e., $\mathtt{SA_k}$) holding a qualified role can easily decrypt the ciphertext $\mathbb{CT}$ using the \textsc{DEC} algorithm as described in Section \ref{Decryption_SR-RBE}.
	\end{remark}
	
	\section{Analysis}
	\label{analysis}
	In this section, security and performance analyses of the proposed scheme are presented. The security analysis shows that the proposed scheme is provably secure against Chosen Plaintext Attack (CPA) under the MDBDH assumption; while the performance analysis presents a comprehensive comparison of the proposed scheme with the closely-related works along with the implementation results.
	
	\subsection{Security Analysis}
	\label{security_analsysis}
	The CPA security of the proposed scheme can be defined by the following theorem and its proof. 
	\begin{theorem}
		\label{cpa}
		If a probabilistic polynomial-time (PPT) adversary $\mathcal{A}$ can win the CPA security game (defined in Section \ref{RBE_security_model}) with non-negligible advantage $\epsilon$, then a PPT simulator $\mathcal{B}$ can be constructed to break the MDBDH assumption with non-negligible advantage $\frac{\epsilon}{2}$.
	\end{theorem}
	\begin{proof}
		In this proof, we show that a simulator $\mathcal{B}$ can be constructed which uses an adversary $\mathcal{A}$ to gain advantage $\frac{\epsilon}{2}$ against the proposed scheme.
		\par 
		The MDBDH challenger $\mathcal{C}$ chooses random numbers $a, b, c, z \in \mathbb{Z}_q^*$ and flips a random coin $\mu\in \{0, 1\}$. It sets $Z= \hat{e}(g, g)^{\frac{ab}{c}}$ if $\mu= 0$ and $Z= \hat{e}(g, g)^z$ otherwise. Afterwards, the challenger $\mathcal{C}$ sends $A= g^{a}, B= g^b, C= g^c, Z$ to the simulator $\mathcal{B}$ and asks it to output $\mu$. Now, the simulator $\mathcal{B}$ acts as a challenger in the rest of the security game.
		\par 
		In the game the simulator $\mathcal{B}$ interacts with an adversary $\mathcal{A}$ as follows: 
		\par
		\textbf{\textsc{Initialization}} Let $\mathbb{S}_{\mathtt{SA}}$ and $\mathbb{S}'_{\mathtt{SA}}$ be the set of all system authorities and corrupted system authorities respectively, where $\mathbb{S}'_{\mathtt{SA}}\subset \mathbb{S}_{\mathtt{SA}}$. Adversary $\mathcal{A}$ submits arbitrary role hierarchies $\mathcal{H}^k$ for each $k\in (\mathbb{S}_{\mathtt{SA}}\setminus\mathbb{S}'_{\mathtt{SA}})$ to the simulator $\mathcal{B}$. It also sends two challenged roles $r^k_{ i}$ and $r^{k'}_{j}$ of any two uncorrupted participating organizations, i.e., $(k, k')\in (\mathbb{S}_{\mathtt{SA}}\setminus\mathbb{S}'_{\mathtt{SA}})$. 
		\par 
		\textbf{\textsc{Setup}} Simulator $\mathcal{B}$ chooses random numbers $\{\{\alpha_{k, i}\}_{\forall r^k_i\in \Psi_k}, \varrho_k, \varphi_k, \zeta_k\in \mathbb{Z}_q^*\}$ for each uncorrupted system authorities, i.e., $\forall k\in (\mathbb{S}_{\mathtt{SA}}\setminus\mathbb{S}'_{\mathtt{SA}})$. It sets $\varrho_k'= \varrho_k\cdot c$ and $\zeta_k'= \zeta_k\cdot c$. The simulator $\mathcal{B}$ then computes $Y_k= \hat{e}(A, g)= \hat{e}(g, g)^{a}, V_k= \hat{e}(C, g)^{\varrho_k}= \hat{e}(g, g)^{\varrho'_k}, C^{\varrho_k}= g^{\varrho'_k}$ and a long term secret $\left(\frac{A}{C^{\varrho_k}}\right)^{\varphi_k}= g^{(a- \varrho_k')\varphi_k}$. It also computes $T_{r^k_{ j}}= C^{\alpha_{k, j}}= g^{c\cdot \alpha_{k, j}}$ for all $r^k_{ j}\in \Psi_k$ and $h_k= C^{\zeta_k}= g^{\zeta_k'}$. Afterwards, simulator $\mathcal{B}$ generates the role public key $\mathbb{PK}^k_{i}= \left< \mathtt{PK^k_{{i}}}, \{\mathtt{AR^k_{j}}\}_{\forall  r^k_{ j}\in \mathbb{A}_{r^k_{ i}}}, r^k_{ i}\right>$ for each $k\in (\mathbb{S}_{\mathtt{SA}}\setminus\mathbb{S}'_{\mathtt{SA}})$, i.e., for each uncorrupted system authorities, where
		\begin{align*}
		\mathtt{PK^k_{{i}}}=& \prod_{ r^k_{ j}\in \bar{\mathbb{A}}_{r^k_{i}}}T_{r^k_{ j}}= g^{c\sum_{r^k_{ j}\in \bar{\mathbb{A}}_{r^k_{i}}}\alpha_{k, j}};
		\mathtt{AR^k_{j}}=& T_{k, j}= g^{c\cdot \alpha_{k, j}}
		\end{align*}
		\par 
		Simulator $\mathcal{B}$ sends $\mathtt{PK_{SA_k}}= \left<q, \mathbb{G}_1, \mathbb{G}_T, \hat{e}, g, H_1, Y_k, V_k, h_k\right>$ along with $\mathbb{PK}^k_{i}, g^{\varrho'_k}$ and $g^{a}$ to the adversary $\mathcal{A}$, where $k\in (\mathbb{S}_{\mathtt{SA}}\setminus\mathbb{S}'_{\mathtt{SA}})$. It also sends the re-encryption key $\mathtt{ReKey_{k, k'}}= \frac{g^{(a- \varrho_k')\varphi_k}}{g^{(a- \varrho'_{k'})}}$ to the adversary $\mathcal{A}$ for each pair of $(k, k')\in (\mathbb{S}_{\mathtt{SA}}\setminus\mathbb{S}'_{\mathtt{SA}})$.
		
		\par 
		\textbf{\textsc{Phase 1}} Adversary $\mathcal{A}$ submits two roles $r^k_{ m}$ and $r^{k'}_m$ so that $r^k_{ m}\notin \mathbb{A}_{r^k_{ i}}$ and $r^{k'}_{ m}\notin \mathbb{A}_{r^{k'}_{j}}$ to the simulator $\mathcal{B}$ in a key generation query, where  $r^k_{ m}$ and $r^{k'}_m$ are the roles from two uncorrupted system authorities, i.e., $(k, k')\in(\mathbb{S}_{\mathtt{SA}}\setminus\mathbb{S}'_{\mathtt{SA}})$. Adversary $\mathcal{A}$ also sends an identity $ID^k_{u}$ to the simulator $\mathcal{B}$ for any $k\in (\mathbb{S}_{\mathtt{SA}}\setminus\mathbb{S}'_{\mathtt{SA}})$.	Simulator $\mathcal{B}$ chooses a random number $\mathfrak{u}_{k, u}\in \mathbb{Z}_q^*$ (private key of the adversary). It then computes $\mathtt{US_{ID^k_u}}= A^{\mathfrak{u}_{k, u}}$. Later on, simulator $\mathcal{B}$ computes public key $\mathtt{Pub_{ID^k_{u}}}$ and role-key $\mathtt{RK^k_{{m, u}}}$ as follows:
		\begin{align*}
		\mathtt{Pub_{ID^k_{u}}}&= g^{\frac{[\mathfrak{u}_{k, u}+ H_1(ID_{k, u})]\varrho_k}{\zeta_k}}\\
		&= g^{\frac{[\mathfrak{u}_{k, u}+ H_1(ID_{k, u})]c\cdot \varrho_k}{c\cdot \zeta_k}}= g^{\frac{[\mathfrak{u}_{k, u}+ H_1(ID_{k, u})]\varrho'_k}{\zeta_k'}}\\
		\mathtt{RK^k_{{m, u}}}&= \left(\mathtt{US_{ID^k_u}}\cdot C^{\varrho_k\cdot H_1(ID_{k, u})}\right)^{\frac{1}{\sum_{r^k_{j}\bar{\mathbb{A}}_{r^k_{ m}}}\alpha_{k, j}}}\\
		&= \left(A^{\mathfrak{u}_{k, u}}\cdot C^{\varrho_k\cdot H_1(ID_{k, u})}\right)^{\frac{1}{\sum_{r^k_{j}\bar{\mathbb{A}}_{r^k_{ m}}}\alpha_{k, j}}}\\
		&= g^{\frac{[a\cdot \mathfrak{u}_{k, u}+ c\cdot\varrho_k\cdot H_1(ID_{k, u})]}{\sum_{r^k_{j}\bar{\mathbb{A}}_{r^k_{ m}}}\alpha_{k, j}}}= g^{\frac{[a\cdot \mathfrak{u}_{k, u}+ \varrho'_k\cdot H_1(ID_{k, u})]}{\sum_{r^k_{j}\bar{\mathbb{A}}_{r^k_{ m}}}\alpha_{k, j}}}
		\end{align*}
		\par 
		Finally, the simulator $\mathcal{B}$ sends $\mathtt{Pub_{ID^k_{u}}}$, $\mathtt{RK^k_{{m, u}}}$ and $\mathfrak{u}_{k, u}$ to the adversary $\mathcal{A}$. 
		\par 
		\textbf{\textsc{Challenge}} When the adversary $\mathcal{A}$ decides that \textbf{\textsc{Phase 1}} is over, he/she submits two equal length messages $\mathtt{K_0}$ and $\mathtt{K_1}$ to the simulator $\mathcal{B}$. The simulator $\mathcal{B}$ selects $b\in \{0, 1\}$ at random and generates a challenged ciphertext $\mathbb{CT}_b= \Big<C_b, C_2, \{C_{3l}| \forall r^k_l\in \mathbb{A}_{r^k_i}\}, C_{k, i}, C'_2, \{C'_{3l}| \forall r^{k'}_l\in \mathbb{A}_{r^{k'}_j}\}, C_{k', i}, r^k_i, r^{k'}_j\Big>$, where 
		\begin{align*} 
		C_b&= \mathtt{K_b}\frac{Z}{\hat{e}(B^{\varrho_k}, g)}	\\
		&= \mathtt{K_b}\frac{\hat{e}(g, g)^{\frac{ab}{c}}}{\hat{e}(g, g)^{b\cdot \varrho_k}}\\
		&= \mathtt{K_b}\frac{\hat{e}(g, g)^{a\cdot r'}}{\hat{e}(g, g)^{\varrho_k\cdot c\cdot \frac{b}{c}}} \text{ \Big(let $r'= \frac{b}{c}$\Big)}\\
		&= \mathtt{K_b}\frac{\hat{e}(g, g)^{a\cdot r'}}{\hat{e}(g, g)^{\varrho'_k\cdot r'}}\\
		C_2&= B^{\zeta_k}= g^{\zeta_k\cdot c\cdot\frac{b}{c}}= g^{\zeta_k'\cdot r'}\\
		C_{3l}&= B^{\alpha_{k, l}}= g^{c\cdot \alpha_{k, l}\cdot \frac{b}{c}}= g^{c\cdot \alpha_{k, l}\cdot r'}\\
		C'_2&= B^{\zeta_{k'}}= g^{\zeta_{k'}\cdot c\cdot\frac{b}{c}}= g^{\zeta'_{k'}\cdot r'}\\
		C_{3l}&= B^{\alpha_{k, l}}= g^{c\cdot \alpha_{k, l}\cdot \frac{b}{c}}= g^{c\cdot \alpha_{k, l}\cdot r'}\\
		C'_{3l}&= B^{\alpha_{k', l}}= g^{c\cdot \alpha_{k', l}\cdot \frac{b}{c}}= g^{c\cdot \alpha_{k', l}\cdot r'}\\
		C_{k, i}&= B^{\sum_{r^k_{z}\in \bar{\mathbb{A}}_{r^k_{ i}}}\alpha_{k, z}}\\
		&= g^{b\sum_{r^k_{z}\in \bar{\mathbb{A}}_{r^k_{ i}}}\alpha_{k, z}}\\
		&= g^{\frac{b}{c}\cdot c\sum_{r^k_{z}\in \bar{\mathbb{A}}_{r^k_{ i}}}\alpha_{k, z}}= g^{r'\cdot c\sum_{r^k_{z}\in \bar{\mathbb{A}}_{r^k_{ i}}}\alpha_{k, z} }\\
	C_{k', j}&= B^{\sum_{r^{k'}_{z}\in \bar{\mathbb{A}}_{r^{k'}_{j}}}\alpha_{k', z}}\\
		&= g^{b\sum_{r^{k'}_{z}\in \bar{\mathbb{A}}_{r^{k'}_{j}}}\alpha_{k', z}}\\
		&= g^{\frac{b}{c}\cdot c\sum_{r^{k'}_{z}\in \bar{\mathbb{A}}_{r^{k'}_{j}}}\alpha_{k', z}}= g^{r'\cdot c\sum_{r^{k'}_{z}\in \bar{\mathbb{A}}_{r^{k'}_{j}}}\alpha_{k', z} }
		\end{align*}
		\par 
			\begin{table}[b]
			\tabcolsep 1.0pt
			\centering
			\caption{NOTATIONS}
			\begin{tabular}{p{2.1cm}p{6cm}}
				\hline
				Notation  & Description
				\\[0.5ex]    \hline
				$|\mathbb{G}_1|, |\mathbb{G}_T|$ & size of an element in $\mathbb{G}_1$ and $\mathbb{G}_T$ respectively  \\\hline
				$T_{exp_{\mathbb{G}_1}}$ &  computation cost of one exponentiation operation on an element of $\mathbb{G}_1$\\\hline
				$T_{exp_{\mathbb{G}_T}}$ &  computation cost of one exponentiation operation on an element of $\mathbb{G}_T$\\\hline
				$T_p$ &  computation cost of one pairing operation\\\hline
				$n_c$ & number of roles associated with a ciphertext \\\hline
				$n_t$ & total number of roles in an organization\\\hline
				$n_r$ & total number of ancestor roles\\\hline
				$n_u$ & total number of users associated with a role \\\hline
			\end{tabular}
			\label{notation2}
		\end{table}
		\begin{table*}[t]
			\centering
			\caption{Functionality Comparison}
			\begin{tabular}{|c|c|c|c|c|c|c|c|}
				\hline
				\multirow{2}{*}{} & \multicolumn{2}{c|}{Organization} & \multirow{2}{*}{Private cloud}&
				\multirow{2}{*}{\parbox{2.5cm}{Provide outsourced decryption?}}& \multirow{2}{*}{User revocation} & \multirow{2}{*}{\parbox{2.5cm}{Re-encryption after revocation}} & \multirow{2}{*}{Revocation controller} \\ \cline{2-3}
				& Single & Multi &  & && &  \\ \hline
				\cite{Zhu2013}	& Yes & No &Not required&No& Yes & Required & Data owner \\ \hline
				\cite{Zhou2013}	& Yes & No &Required&Yes& Yes & Not required & System Administrator \\ \hline
				SO-RBE & Yes & No& Required&Yes& Yes & Not required & System Administrator \\ \hline
				MO-RBE & Yes & Yes& Required&Yes& Yes & Not required & System Administrator \\ \hline
			\end{tabular}
			\label{functionality_RBE}
		\end{table*}
		\begin{table*}[t]
			\centering
			\caption{Storage and Computation Costs Comparison }
			\begin{tabular}{|l|p{2.5cm}|l|p{1.7cm}|p{2.3cm}|p{2.7cm}|p{2.5cm}|}
				\hline
				\multirow{2}{*}{} & \multicolumn{3}{c|}{Storage overhead} & \multicolumn{3}{c|}{Computation overhead} \\ \cline{2-7} 
				& \multicolumn{1}{c|}{Ciphertext size} & \multicolumn{1}{c|}{User key size} & \multicolumn{1}{c|}{Master secret size} & \multicolumn{1}{c|}{Encryption} & \multicolumn{1}{c|}{Decryption} & User revocation \\ \hline
				\cite{Zhu2013}	& $(2+ n_c)|\mathbb{G}_1|+ |\mathbb{G}_T|$  & $2|\mathbb{G}_1|+ |\mathbb{Z}_q^*|$  & $(n_t+ 1)|\mathbb{Z}_q^*|+ |\mathbb{G}_1|$ & $(2+ n_c)T_{exp_{\mathbb{G}_1}}+ T_{exp_{\mathbb{G}_T}}$ & $2T_p$ &  $T_{exp_{\mathbb{G}_1}}+ T_{exp_{\mathbb{G}_T}}$\\ \hline
				\cite{Zhou2013}	& $3|\mathbb{G}_1|+ |\mathbb{G}_T|$ & $|\mathbb{G}_1|$ & $3|\mathbb{Z}_q^*|$ & $3T_{exp_{\mathbb{G}_1}}+ T_{exp_{\mathbb{G}_T}}$ & $(n_r+n_u)T_{exp_{\mathbb{G}_1}}+ 2T_{exp_{\mathbb{G}_T}}+ 5T_p$ &  $6T_{exp_{\mathbb{G}_1}}+ T_{exp_{\mathbb{G}_T}}+ 2T_p$\\ \hline
				SO-RBE	& $(1+ n_c)|\mathbb{G}_1|+ |\mathbb{G}_T|$ &  $|\mathbb{G}_1|+ |\mathbb{Z}_q^*|$ & $(4+ n_t)|\mathbb{Z}_q^*|$ & $(1+ n_c)T_{exp_{\mathbb{G}_1}}+ T_{exp_{\mathbb{G}_T}}$ & $T_{exp_{\mathbb{G}_1}}+ 2T_{exp_{\mathbb{G}_T}}+ 2T_p$ & Negligible \\ \hline
				MO-RBE	& $(2+ n_c)|\mathbb{G}_1|+ |\mathbb{G}_T|$ & $|\mathbb{G}_1|+ |\mathbb{Z}_q^*|$ & $(4+ n_t)|\mathbb{Z}_q^*|$ & $(2+ n_c)T_{exp_{\mathbb{G}_1}}+ T_{exp_{\mathbb{G}_T}}$  & $4T_{exp_{\mathbb{G}_1}}+ 2T_{exp_{\mathbb{G}_T}}+ 3T_p$  & Negligible \\ \hline
			\end{tabular}
			\label{storage_computation_RBE}
		\end{table*}
		Simulator $\mathcal{B}$ sends $\mathbb{CT}_b$ to the adversary $\mathcal{A}$. 
		
		\par 
		\textbf{\textsc{- Phase 2}} Same as \textbf{\textsc{Phase 1}}.
		\par
		\textbf{\textsc{- Guess}} The adversary $\mathcal{A}$ guesses a bit $b'$ and sends to the simulator $\mathcal{B}$. If $b'=b$ then the adversary $\mathcal{A}$ wins CPA game; otherwise it fails. If $b'= b$, simulator $\mathcal{B}$ answers ``MDBDH'' in the game (i.e. outputs $\mu= 0$); otherwise $\mathcal{B}$ answers ``random'' (i.e. outputs $\mu= 1$).
		\par 
		If $Z= \hat{e}(g, g)^{z}$; then $C_{b}$ is completely random from the view of the adversary $\mathcal{A}$. So, the received ciphertext $\mathbb{CT}$ is not compliant to the game (i.e. invalid ciphertext). Therefore, the adversary $\mathcal{A}$ chooses $b'$ randomly. Hence, probability of the adversary $\mathcal{A}$ for outputting $b'= b$ is $\frac{1}{2}$.
		\par 
		
		If $Z= \hat{e}(g, g)^{\frac{ab}{c}}$, then adversary $\mathcal{A}$ receives a valid ciphertext. The adversary $\mathcal{A}$ wins the CPA game with non-negligible advantage $\epsilon$ (according to Theorem \ref{cpa}).  So, the probability of outputting $b'= b$ for the adversary $\mathcal{A}$ is $\frac{1}{2}+ \epsilon$, where probability $\epsilon$ is for guessing that the received ciphertext is valid and probability $\frac{1}{2}$ is for guessing whether the valid ciphertext $C_{b}$ is related to $\mathtt{K_0}$ or $\mathtt{K_1}$.
		\par 
		Therefore, the overall advantage $Adv_{\mathcal{A}}$ of the simulator $\mathcal{B}$ is $\frac{1}{2}(\frac{1}{2}+ \epsilon+ \frac{1}{2})- \frac{1}{2}= \frac{\epsilon}{2}$.
	\end{proof}

	\subsection{Performance Analysis}
	\label{performance_analysis}
	The performance analysis of the proposed scheme is presented in two parts, namely comprehensive analysis and implementation results. In the comprehensive analysis, a comparison is made between the proposed scheme (i.e., SO-RBE and MO-RBE mechanisms) and closely-related works Zhou \emph{et al.}'s scheme \cite{Zhou2013} and Zhu \emph{et al.}'s scheme \cite{Zhu2013} in terms of functionalities, storage and computation overhead. For this purpose, the same security level for the computation of cryptographic algorithms is considered in all studied schemes. In the implementation results, the proposed scheme is compared with notable works Zhou \emph{et al.}'s scheme \cite{Zhou2013} and Zhu \emph{et al.}'s scheme \cite{Zhu2013} in terms of computation time. The notations used in the following subsections of this paper are shown in Table \ref{notation2}.
	
	\subsubsection{Comprehensive Analysis}
	Table \ref{functionality_RBE} shows a comparison of the proposed scheme with the closely-related works Zhou \emph{et al.}'s scheme \cite{Zhou2013} and Zhu \emph{et al.}'s scheme \cite{Zhu2013} in terms of some essential functionalities. From the table, it can be observed that the proposed scheme provides access control for both single-organization and multi-organization cloud storage systems; while the other schemes provide access control only for the single-organization cloud storage system. Unlike \cite{Zhu2013}, the proposed scheme (i.e., SO-RBE and MO-RBE) uses private cloud to keep user and role secrets. Moreover, the private cloud in the proposed MO-RBE mechanism is used for providing assistance to the public cloud during outsourced decryption. Also, unlike \cite{Zhu2013}, the proposed SO-RBE and MO-RBE mechanisms support outsourced decryption property which reduces computation overhead on the user's side. In \cite{Zhou2013} and \cite{Zhu2013}, users are revoked by updating all the role public keys and by re-encrypting all the ciphertexts related with the revoked users respectively. This may increase overhead on the system if frequent user revocations occur. In \cite{Zhu2013}, the data owners control user revocation by embedding revoked user identities during encryption process. But, it requires knowledge of the revoked users during the encryption process which may not be feasible always. Moreover, it increases computation overhead on the data owner side if the frequent user revocation occurs. On the other hand, in the proposed scheme, users are revoked from the system by invalidating the public keys of the users. Thus, the proposed scheme takes very less computation overhead to revoke users compared with the other closely-related schemes. 
	\par 
	Table \ref{storage_computation_RBE} shows storage and computation overhead comparison of the proposed scheme (i.e., SO-RBE and MO-RBE mechanisms) with Zhou \emph{et al.}'s scheme \cite{Zhou2013} and Zhu \emph{et al.}'s scheme \cite{Zhu2013}. The comparison is shown in asymptotic upper bound in the worst cases. The storage overhead comparison is done in terms of group element size (i.e., $|\mathbb{G}_1|, |\mathbb{G}_T|$ and $|\mathbb{Z}^*_q|$); while the computation cost comparison is done in terms of number of exponentiation and pairing operations (i.e., $T_{exp_{\mathbb{G}_1}}, T_{exp_{\mathbb{G}_T}}$ and $T_p$). It is to be noted that the computation cost of group element multiplications and hash operations are negligible compared with the pairing and exponentiation operations (the reason is explained in Section \ref{implementation_RBE}). Therefore, the group element multiplications and hash operations are ignored in the rest of the comparisons. For the storage overhead comparison, ciphertext size, user secret key size and master secret key size are considered; while the encryption, decryption and user revocation costs are considered for comparing computation overhead. The encryption cost of the proposed scheme and \cite{Zhu2013} mainly depend on the number of ancestor roles associated with a ciphertext, as both the schemes need to compute one exponentiation operation per ancestor role during encryption process. While, in \cite{Zhou2013}, the encryption process needs constant four exponentiation operations. It is to be noted that encryption operation for a data is performed only once. In \cite{Zhu2013}, the decryption processing cost is constant which takes only two pairing operations, and the pairing operations are performed by the users. In  \cite{Zhou2013}, the decryption processing cost mainly depends on the number of ancestor roles and the number of users associated with a role. On the other hand, in the proposed scheme, a user needs to perform only three exponentiation operations in both SO-RBE and MO-RBE mechanisms during decryption process, and the other cryptographic operations are performed by the public and private clouds. Hence, the decryption processing cost of the proposed scheme is comparable with \cite{Zhu2013}; while it takes considerably less cost compared with \cite{Zhou2013}. In \cite{Zhu2013}, user revocation operation requires two exponentiation operations per ciphertext; while in \cite{Zhou2013}, user revocation operation takes five exponentiation operations and one pairing operation. On the other hand, in the proposed scheme, the users are revoked from the system by removing or invalidating their public keys which takes considerably less computation cost.
	\begin{table}[!t]
		\caption{Computation Time (in Milliseconds) of Elementary Cryptographic Operations}
		\begin{tabular}{|p{.5cm}|l|l|p{.7cm}|l|c|p{.5cm}|}
			\hline
			\multirow{2}{*}{} & \multicolumn{2}{p{1.5cm}|}{Exponentiation} & \multirow{2}{*}{Pairing} & \multicolumn{2}{p{2.4cm}|}{Group multiplication} & \multirow{2}{*}{Hash} \\ \cline{2-3} \cline{5-6} 
			& \multicolumn{1}{c|}{$\mathbb{G}_1$} & \multicolumn{1}{c|}{$\mathbb{G}_T$} &  & \multicolumn{1}{c|}{$\mathbb{G}_1$} & \multicolumn{1}{c|}{$\mathbb{G}_T$} &  \\ \hline
			\multicolumn{1}{|c|}{\begin{tabular}[c]{@{}c@{}}Commodity \\ Laptop\end{tabular}} & $2.062$ & $0.126$ & $1.292$ & $0.008$ & $0.002$ & $0.003$ \\ \hline
			\multicolumn{1}{|c|}{Workstation} & $1.153$ & $0.091$ & $0.645$ & $0.005$ & $0.001$ & $0.002$ \\ \hline
		\end{tabular}
		\label{time_comp}
	\end{table}
	\begin{figure}[t]
		\centering
		\scalebox{3}{\includegraphics[width=3cm, height=2.5cm]{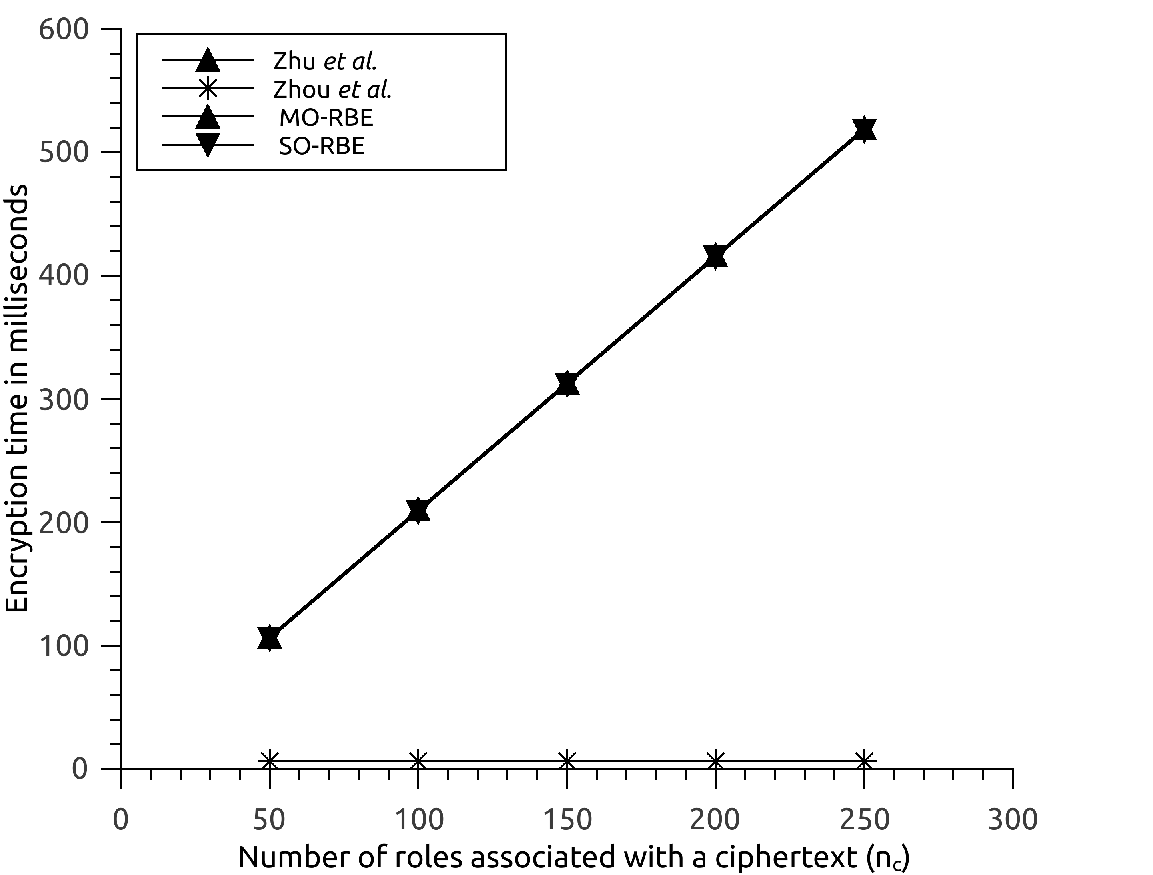}} 
		\caption{Encryption Time Comparison of SO-RBE, MO-RBE with Zhou \emph{et al.}'s scheme \cite{Zhou2013} and Zhu \emph{et al.}'s scheme \cite{Zhu2013}}
		\label{encryption_time_RBE}
	\end{figure}
	\begin{figure}[t]
		\centering
		\scalebox{3}{\includegraphics[width=3cm, height=2.3cm]{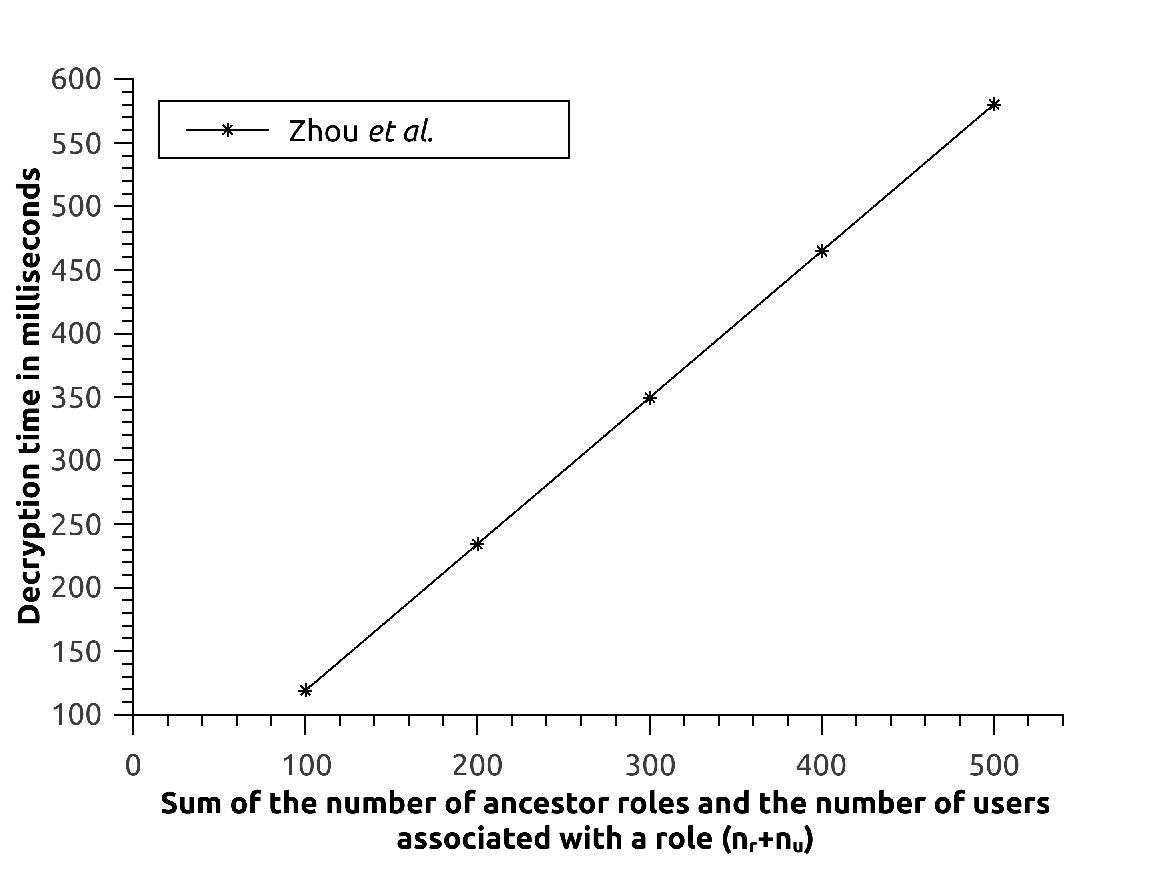}} 
		\caption{Decryption Time in Zhou \emph{et al.}'s scheme \cite{Zhou2013}}
		\label{decryption_time_Zhou}
	\end{figure}
	\subsubsection{Implementation Results}
	\label{implementation_RBE}
	The proposed scheme as well as two closely-related constructions Zhou \emph{et al.}'s scheme \cite{Zhou2013} and Zhu \emph{et al.}'s scheme \cite{Zhu2013} are implemented using Pairing-Based Cryptography (PBC) library \cite{PBC}. The elementary cryptographic operations that are performed by the data owners and users are implemented using a commodity Laptop Computer having Ubuntu 17.10 (64-bit) operating system and having 2.4GHz Core i3 processor with 4GB memory. The elementary cryptographic operations that are performed by public and private clouds are implemented using a workstation having Ubuntu 17.10 (64-bit) operating system and having 3.5 GHz Intel(R) Xeon(R) CPU E5-2637 v4 processor with 16 GB memory. For the implementation purpose, type A elliptic curve with embedding degree $2$ of $160$-bit group order is used, which provides an equivalent $1024$-bit discrete log security. The implementation results of the elementary mathematical functions like exponentiation operations on $\mathbb{G}_1, \mathbb{G}_T$, pairing operation, group element multiplication operations on $\mathbb{G}_1, \mathbb{G}_T$,  and hash operation are shown in Table \ref{time_comp}. From the Table \ref{time_comp}, it is observed that the computation time of group element multiplications on $\mathbb{G}_1, \mathbb{G}_T$ and hash operation are negligible compared with the cryptographic exponentiation and pairing operations. 
	\par
	Figure \ref{encryption_time_RBE} shows the encryption time of the proposed scheme, i.e., SO-RBE and MO-RBE along with Zhou \emph{et al.}'s scheme \cite{Zhou2013} and Zhu \emph{et al.}'s scheme \cite{Zhu2013}. From the figure, it is observed that the encryption time of the proposed scheme and \cite{Zhu2013} linearly increase with the number of roles associated with a ciphertext; while the encryption time in \cite{Zhou2013} takes constant time, as in \cite{Zhou2013} the encryption process takes constant four exponentiation operations. Figure \ref{decryption_time_Zhou} shows decryption time in Zhou \emph{et al.}'s scheme \cite{Zhou2013}. As the decryption time in \cite{Zhou2013} increases with the number of ancestor roles associated with a ciphertext and the number of users associated with a role, the decryption time in \cite{Zhou2013} linearly increases with the number of ancestor roles and number of users associated with a role. On the other hand, the decryption time of the proposed SO-RBE, MO-RBE mechanisms and \cite{Zhu2013} take constant time, approximately $4.60$ milliseconds (ms), $9.675$ ms and $3.58$ ms respectively. Hence, it can be concluded that the decryption time of the proposed scheme take considerably less than of \cite{Zhou2013}, while it is comparable with  \cite{Zhu2013}.
	
	\section{Conclusion}
	\label{conclusion}
	In this paper, a novel cryptographic role-based access control scheme has been proposed for both single and multi-organization cloud storage systems. The proposed scheme employs a Role-Based Encryption (RBE) method to enforce RBAC access policies in encrypted data that enables only the authorized users with qualified roles are able to decrypt. It delegates costly cryptographic operations in the decryption process to the cloud that reduces the overhead on the user side. It also enables the system administrator to revoke a user from the system by revoking/invalidating user's public key with minimal overhead. In addition, we have demonstrated that the proposed scheme is provably secure against CPA under the MDBDH cryptographic assumption. Furthermore, we have shown that the performance of the proposed scheme is better compared with the previous schemes. 
	\par
	Dynamic change of roles in the role hierarchy poses challenges in role-based access control systems. Design of efficient cryptographic role-based access control scheme supporting dynamic changes in roles will be one of the areas of future work. Another area of interest involves the development of trust models for a multi-organization cloud environment. In this context, the model such as the one in \cite{Zhou2015} can be used to extend our proposed scheme to enhance the trust of the data owners on the multi organization system authorities.

	\bibliographystyle{unsrt}

\begin{thebibliography}{10}

\bibitem{Azure}
Microsoft Azure Storage Service.
  \url{https://azure.microsoft.com/en-in/services/storage/}[Online accessed:
  8-Nov.-2019].

\bibitem{AmazonS3}
Amazon Simple Storage Service. \url{http://aws.amazon.com/s3/} [Online
  accessed: 8-Nov.-2019].

\bibitem{GoogleCloudStorage}
Google Cloud Storage. \url{https://cloud.google.com/storage/}[Online accessed:
  8-Nov.-2019].

\bibitem{Armbrust2010}
M.~Armbrust \emph{et al.}
\newblock {A View of Cloud Computing}.
\newblock {\em Commun. ACM}, 53(4):50--58, April 2010.

\bibitem{Kamara2010}
S.~Kamara and K.~Lauter.
\newblock Cryptographic cloud storage.
\newblock In {\em Proc. of the 14th International Conference on Financial
  Cryptograpy and Data Security}, FC'10, pages 136--149, 2010.

\bibitem{Shin2017}
Y.~Shin, D.~Koo, and J.~Hur.
\newblock {A Survey of Secure Data Deduplication Schemes for Cloud Storage
  Systems}.
\newblock {\em ACM Comp. Sur.}, 49(4):74:1--74:38, Jan. 2017.

\bibitem{Zhou2013}
L.~Zhou, V.~Varadharajan, and M.~Hitchens.
\newblock Achieving secure role-based access control on encrypted data in cloud
  storage.
\newblock {\em IEEE Trans. Info. For. and Sec.}, 8(12):1947--1960, Dec 2013.

\bibitem{McAfee}
McAfee.
\newblock Navigating a cloudy sky: Practical guidance and the state of cloud
  security.
\newblock White paper, 2018.

\bibitem{Verizon}
Verizon.
\newblock {Verizon Data Breach Investigations Report}.
\newblock 2017.
\newblock \url{https://enterprise.verizon.com/resources/reports/2017_dbir.pdf}
  [Online accessed: 8-Nov.-2019].

\bibitem{Ferraiolo1999}
D.~F. Ferraiolo, J.~F. Barkley, and D.~R. Kuhn.
\newblock A role-based access control model and reference implementation within
  a corporate intranet.
\newblock {\em ACM Trans. on Info. and Syst. Sec.}, 2(1):34--64, Feb. 1999.

\bibitem{Ferraiolo1992}
R.~Kuhn D.~Ferraiolo.
\newblock Role-based access controls.
\newblock In {\em Proc. of the 15th National Computer Security Conference},
  pages 554--563, Oct. 1992.
\newblock
  \url{https://csrc.nist.gov/publications/detail/conference-paper/1992/10/13/role-based-access-controls}.

\bibitem{Zhou2011}
L.~Zhou, V.~Varadharajan, and M.~Hitchens.
\newblock {Enforcing Role-Based Access Control for Secure Data Storage in the
  Cloud}.
\newblock {\em The Computer Journal}, 54(10):1675--1687, Oct. 2011.

\bibitem{Zhu2013}
Y.~Zhu, G.~Ahn, H.~Hu, D.~Ma, and S.~Wang.
\newblock Role-based cryptosystem: A new cryptographic {RBAC} system based on
  role-key hierarchy.
\newblock {\em IEEE Trans. Info. For. and Sec.}, 8(12):2138--2153, Dec 2013.

\bibitem{Perez2017}
J.~M.~Marín Pérez, G.~M. Pérez, and A.~F.~Skarmeta Gomez.
\newblock {SecRBAC: Secure data in the Clouds}.
\newblock {\em IEEE Trans. on Serv. Comp.}, 10(5):726--740, Sept 2017.

\bibitem{Akl1983}
S.~G. Akl and P.~D. Taylor.
\newblock Cryptographic solution to a problem of access control in a hierarchy.
\newblock {\em ACM Trans. on Comp. Syst.}, 1(3):239--248, Aug. 1983.

\bibitem{Chang2004}
C.~C. Chang, I.~C. Lin, H.~M. Tsai, and H.~H. Wang.
\newblock A key assignment scheme for controlling access in partially ordered
  user hierarchies.
\newblock In {\em Proc. of the 18th International Conference on Advanced
  Information Networking and Applications - Volume 2}, AINA '04, pages 376--,
  2004.

\bibitem{Atallah2009}
M.~J. Atallah, M.~Blanton, N.~Fazio, and K.~B. Frikken.
\newblock Dynamic and efficient key management for access hierarchies.
\newblock {\em ACM Trans. on Info. and Syst. Sec.}, 12(3):18:1--18:43, Jan.
  2009.

\bibitem{Lin2011}
Y.L. Lin and C.~L. Hsu.
\newblock Secure key management scheme for dynamic hierarchical access control
  based on ecc.
\newblock {\em J. of Syst. and Soft.}, 84(4):679 -- 685, 2011.

\bibitem{Tang2016b}
S.~Tang \emph{et al.}
\newblock Achieving simple, secure and efficient hierarchical access control in
  cloud computing.
\newblock {\em IEEE Trans. on Comp.}, 65(7):2325--2331, July 2016.

\bibitem{Chen2017}
Y.~R. Chen and W.~G. Tzeng.
\newblock Hierarchical key assignment with dynamic read-write privilege
  enforcement and extended ki-security.
\newblock In {\em Proc. of the Applied Cryptography and Network Security},
  pages 165--183, 2017.

\bibitem{Pareek2018}
G.~Pareek and B.~R. Purushothama.
\newblock Efficient strong key indistinguishable access control in dynamic
  hierarchies with constant decryption cost.
\newblock In {\em Proc. of the 11th International Conference on Security of
  Information and Networks}, SIN '18, pages 10:1--10:7, 2018.

\bibitem{Gentry2002}
C.~Gentry and A.~Silverberg.
\newblock {Hierarchical ID-Based Cryptography}.
\newblock In {\em Proc. of the 8th International Conference on the Theory and
  Application of Cryptology and Information Security: Advances in Cryptology},
  ASIACRYPT '02, pages 548--566, 2002.

\bibitem{Boneh2005}
D.~Boneh, X.~Boyen, and E.-J. Goh.
\newblock {Hierarchical Identity Based Encryption with Constant Size
  Ciphertext}.
\newblock In {\em Proc. of the 24th Annual International Conference on Theory
  and Applications of Cryptographic Techniques}, EUROCRYPT'05, pages 440--456,
  2005.

\bibitem{Goyal2006}
V.~Goyal, O.~Pandey, A.~Sahai, and B.~Waters.
\newblock Attribute-based encryption for fine-grained access control of
  encrypted data.
\newblock In {\em Proc. of the 13th ACM Conference on Computer and
  Communications Security}, CCS '06, pages 89--98, 2006.

\bibitem{Sahai2005}
A.~Sahai and B.~Waters.
\newblock Fuzzy identity-based encryption.
\newblock In {\em Proc. of the 24th Annual International Conference on Theory
  and Applications of Cryptographic Techniques}, EUROCRYPT'05, pages 457--473,
  2005.

\bibitem{Yub}
S.~Yu \emph{et al.}
\newblock {Achieving Secure, Scalable, and Fine-grained Data Access Control in
  Cloud Computing}.
\newblock In {\em Proc. of the 29th Conference on Information Communications},
  INFOCOM'10, pages 534--542, 2010.

\bibitem{Bethencourt2007}
J.~Bethencourt, A.~Sahai, and B.~Waters.
\newblock Ciphertext-policy attribute-based encryption.
\newblock In {\em Proc. of the IEEE Symposium on Security and Privacy}, SP '07,
  pages 321--334, May 2007.

\bibitem{Chase2007}
M.~Chase.
\newblock {Multi-authority Attribute Based Encryption}.
\newblock In {\em Proc. of the 4th Conference on Theory of Cryptography},
  TCC'07, pages 515--534, 2007.

\bibitem{Yang2014}
K.~{Yang} and X.~{Jia}.
\newblock {Expressive, Efficient, and Revocable Data Access Control for
  Multi-Authority Cloud Storage}.
\newblock {\em IEEE Trans. on Parallel and Dist. Syst.}, 25(7):1735--1744, July
  2014.

\bibitem{Chuangui2018}
C.~Ma, A.~Ge, and J.~Zhang.
\newblock {Fully Secure Decentralized Ciphertext-Policy Attribute-Based
  Encryption in Standard Model}.
\newblock In {\em Information Security and Cryptology}, pages 427--447, 2019.

\bibitem{FUGKEAW2018}
S.~Fugkeaw and H.~Sato.
\newblock {Scalable and secure access control policy update for outsourced big
  data}.
\newblock {\em Future Gener. Comput. Syst.}, 79:364 -- 373, 2018.

\bibitem{Wang2016}
F.~Wang, Z.~Liu, and C.~Wang.
\newblock Full secure identity-based encryption scheme with short public key
  size over lattices in the standard model.
\newblock {\em International J. of Comp. Math.}, 93(6):854--863, 2016.

\bibitem{Yang2013}
K.~Yang \emph{et al.}
\newblock {DAC-MACS: Effective Data Access Control for Multi-authority Cloud
  Storage Systems}.
\newblock {\em IEEE Trans. Info. For. and Sec.}, 8(11):1790--1801, Nov 2013.

\bibitem{Xue2017}
K.~{Xue} \emph{et al.}
\newblock {RAAC: Robust and Auditable Access Control With Multiple Attribute
  Authorities for Public Cloud Storage}.
\newblock {\em IEEE Trans. Info. For. and Sec.}, 12(4):953--967, April 2017.

\bibitem{WANG2011}
G.~Wang, Q.~Liu, J.~Wu, and M.~Guo.
\newblock {Hierarchical attribute-based encryption and scalable user revocation
  for sharing data in cloud servers}.
\newblock {\em Computers \& Security}, 30(5):320 -- 331, 2011.

\bibitem{Wan2012}
Z.~{Wan}, J.~{Liu}, and R.~H. {Deng}.
\newblock {HASBE: A Hierarchical Attribute-Based Solution for Flexible and
  Scalable Access Control in Cloud Computing}.
\newblock {\em IEEE Trans. on Info. For. and Sec.}, 7(2):743--754, April 2012.

\bibitem{Wei2019}
J.~{Wei} \emph{et al.}
\newblock {RS-HABE: Revocable-storage and Hierarchical Attribute-based Access
  Scheme for Secure Sharing of e-Health Records in Public Cloud}.
\newblock {\em IEEE Trans. on Dep. and Sec. Comp.}, pages 1--1, 2019.

\bibitem{Hur2011}
J.~{Hur} and D.~K. {Noh}.
\newblock {Attribute-Based Access Control with Efficient Revocation in Data
  Outsourcing Systems}.
\newblock {\em IEEE Trans. Parallel Distrib. Syst.}, 22(7):1214--1221, July
  2011.

\bibitem{Yu2010}
S.~Yu \emph{et al.}
\newblock Attribute based data sharing with attribute revocation.
\newblock In {\em Proc. of the 5th ACM Symposium on Information, Computer and
  Communications Security}, ASIACCS '10, pages 261--270, 2010.

\bibitem{Hur2013}
J.~{Hur}.
\newblock {Improving Security and Efficiency in Attribute-Based Data Sharing}.
\newblock {\em IEEE Trans. on Knowl. and Data Eng.}, 25(10):2271--2282, Oct
  2013.

\bibitem{PBC}
PBC (Pairing-Based Cryptography) library. \url{http://crypto.stanford.edu/pbc/}
  [Online accessed: 8-Nov.-2019].

\bibitem{Zhou2015}
L.~{Zhou}, V.~{Varadharajan}, and M.~{Hitchens}.
\newblock {Trust Enhanced Cryptographic Role-Based Access Control for Secure
  Cloud Data Storage}.
\newblock {\em IEEE Trans. Info. For. and Sec.}, 10(11):2381--2395, Nov 2015.

\end{thebibliography}

\end{document}